\documentclass[a4paper, 12pt, reqno]{amsart}

\usepackage{amsmath,amsfonts,amsthm,amssymb,here}
\usepackage{hyperref}
\usepackage{amsmath}
\usepackage{graphicx}
\usepackage{natbib}

\usepackage[utf8x]{inputenc}
\usepackage[T1]{fontenc}
\usepackage{lmodern}
\usepackage{textcomp}

\begin{document}

\newtheorem{theo}{Theorem}[section]
\newtheorem{lemme}[theo]{Lemma}
\newtheorem{cor}[theo]{Corollary}
\newtheorem{defi}[theo]{Definition}
\newtheorem{prop}{Proposition}
\newtheorem{problem}[theo]{Problem}
\newtheorem{remarque}[theo]{Remark}
\newtheorem{claim}[theo]{Claim}
\newcommand{\beq}{\begin{eqnarray}}
\newcommand{\enq}{\end{eqnarray}}
\newcommand{\be}{\begin{eqnarray*}}
	\newcommand{\en}{\end{eqnarray*}}
\newcommand{\ben}{\begin{eqnarray*}}
	\newcommand{\enn}{\end{eqnarray*}}
\newcommand{\Td}{\mathbb T^d}
\newcommand{\Rd}{\mathbb R^n}
\newcommand{\R}{\mathbb R}
\newcommand{\N}{\mathbb N}
\newcommand{\Sn}{\mathbb S}
\newcommand{\Zd}{\mathbb Z^d}
\newcommand{\Linf}{L^{\infty}}
\newcommand{\dt}{\partial_t}
\newcommand{\Dt}{\frac{d}{dt}}
\newcommand{\Dtt}{\frac{d^2}{dt^2}}
\newcommand{\demi}{\frac{1}{2}}
\newcommand{\vf}{\varphi}
\newcommand{\epu}{_{\epsilon}}
\newcommand{\ep}{^{\epsilon}}
\newcommand{\bfi}{{\mathbf \Phi}}
\newcommand{\bpsi}{{\mathbf \Psi}}
\newcommand{\bx}{{\mathbf x}}
\newcommand{\bX}{{\mathbf X}}

\newcommand{\ds}{\displaystyle}
\newcommand {\g}{\`}
\newcommand{\E}{\mathbb E}
\newcommand{\Q}{\mathbb Q}
\newcommand{\PP}{\mathbb P}
\newcommand{\1}{\mathbb I}
\let\cal=\mathcal
\newcommand{\indicator}[1]{\,\mathbf{1}_{#1}}

\title{Robustness of mathematical models and technical analysis strategies}
\maketitle
\begin{center}
Ahmed Bel Hadj Ayed\footnotemark[1]$^{,}$\footnotemark[2], Gr\'egoire Loeper \footnotemark[3], Fr\'ed\'eric Abergel \footnotemark[1] 
\end{center}
\begin{abstract}
The aim of this paper is to compare the performances of the optimal strategy under parameters mis-specification and of a technical analysis trading strategy. The setting we consider is that of a stochastic asset price model where the trend follows an unobservable Ornstein-Uhlenbeck process. For both strategies, we provide the asymptotic expectation of the logarithmic return as a function of the model parameters. Finally, numerical examples find that an investment strategy using the cross moving averages rule is more robust than the optimal strategy under parameters mis-specification.
\end{abstract}
\footnotetext[1]{Chaire of quantitative finance, laboratory MAS, CentraleSup\'elec}
\footnotetext[2]{BNP Paribas Global Markets}
\footnotetext[3]{Monash University}

\section*{Introduction}
There exist three principal approaches for investments in financial markets (see \cite{TalayTechnical}).
The first one is based on fundamental economic principles (see \cite{economic} for details). The second one is called the technical analysis approach and uses the historical prices and volumes (see \cite{tec1},\cite{tec2} and \cite{tec3} for details).
The third one is the use of mathematical models and was introduced in \cite{Merton}. He assumed that the risky asset follows a geometric Brownian motion and derived the optimal investment rules for an investor maximizing his expected utility function. Several generalisations of this problem are possible (see \cite{Karatzas}, \cite{Brendle}, \cite{LaknerStrat}, \cite{Sass}, or  \cite{rieder} for example) but all these models are confronted to the calibration problem. In \cite{Ahmed}, the authors assess the feasibility of forecasting trends modeled by an unobserved mean-reverting diffusion. They show that, due to a weak signal-to-noise ratio, a bad calibration is very likely. Using the same risky asset model, \cite{Zhou} analyse the performance of a technical analysis strategy based on a geometric moving average rule.  
In  \cite{TalayTechnical}, the authors assume that the drift is an unobservable constant piecewise process jumping at an unknown time. They provide the performance of the optimal trading strategy under parameters mis-specification and compare this strategy to a technical analysis investment based on a simple moving average rule with Monte Carlo simulations. 

In this paper, we consider a stochastic asset price model where the trend is an unobservable Ornstein Uhlenbeck process. The purpose of this work is to characterize and to compare the performances of the optimal strategy under parameters mis-specification and of a cross moving average strategy.

The paper is organized as follows: the first section presents the model, recalls some results from filtering theory and rewrites the Kalman filter estimator as a corrected exponential average.

In the second section, the optimal trading strategy under parameters mis-specification is investigated. For this portfolio, the stochastic differential equation of the logarithmic return is found. Using this result, we provide, in closed form, the asymptotic expectation of the logarithmic return as a function of the signal-to-noise-ratio and of the trend mean reversion speed. We close this section by giving conditions on the model and the strategy parameters that guarantee a positive asymptotic expected logarithmic return and the existence of an optimal duration.

In the third section, we consider a cross moving average strategy. For this portfolio, we also provide the stochastic differential equation of the logarithmic return. We close this section by giving, in closed form, the asymptotic expectation of the logarithmic return as a function of the model parameters.

In the fourth section, numerical examples are performed. First, the best durations of the Kalman filter and of the optimal strategy under parameters  mis-specification are illustrated over several trend regimes. We then compare the performances of a cross moving average strategy and of a classical optimal strategy used in the industry (with a duration $\tau=1$ year) over several theoretical regimes. We also compare these performances under Heston's stochastic volatility model using Monte Carlo simulations. These examples show that the technical analysis approach is more robust than the optimal strategy under parameters mi-specification. We close this study by confirming this conclusion with empirical tests based on real data.

\section{Setup}
This section begins by presenting the model, which corresponds to an unobserved mean-reverting diffusion. After that, we reformulate this model in a completely observable environment (see \cite{Lipster1} for details). This setting introduces the conditional expectation of the trend, knowing the past observations. Then, we recall the asymptotic continuous time limit of the Kalman filter and we rewrite this estimator as a corrected exponential average.
\subsection{The model}
Consider a financial market living on a stochastic basis $( \Omega , \mathcal{F} ,\mathbf{F}, \mathbb{P} )$, where  $\mathbf{F}=\left\lbrace  \mathcal{F}_{t}, t \geqslant 0 \right\rbrace$ is the natural filtration associated to a two-dimensional (uncorrelated) Wiener process $(W^{S},W^{\mu})$, and $\mathbb{P}$ is the objective probability measure. The dynamics of the risky asset $S$ is given by
\begin{eqnarray}\label{Model3}
\frac{dS_{t}}{S_{t}}&=&\mu_{t}dt+\sigma_{S} dW_{t}^{S},\\
d\mu_{t}&=&-\lambda\mu_{t}dt+\sigma_{\mu}dW_{t}^{\mu}\label{Model23},
\end{eqnarray}
with $\mu_{0}=0$. We also assume that $\left( \lambda,\sigma_{\mu},\sigma_{S}\right) \in  \mathbb{R}_+^{*}\times\mathbb{R}_+^{*}\times\mathbb{R}_+^{*}$. The parameter $\lambda$ is called the trend mean reversion speed. Indeed, $\lambda$ can be seen as the "force" that pulls the trend back to zero. 
Denote by $\mathbf{F}^{S}=\left\lbrace  \mathcal{F}_{t}^S \right\rbrace$ be the natural filtration associated to the price process $S$. An important point is that only $\mathbf{F}^{S}$-adapted processes are observable, which implies that agents in this market do not observe the trend $\mu$.
\subsection{The observable framework}
As stated above, the agents can only observe the stock price process $S$. Since, the trend $\mu$ is not $F^S$-measurable, the agents do not observe it directly. Indeed, the model \hyperref[Model3]{(1)-(2)} corresponds to a system with partial information.  The following proposition gives a representation of the model \hyperref[Model3]{(1)-(2)} in an observable framework (see \cite{Lipster1} for details). 
\begin{prop}\label{ObservableFramework3}
The dynamics of the risky asset $S$ is also given by
\begin{eqnarray}\label{ObservableSDE3}
\frac{dS_t}{S_t} & = & E\left[  \mu_t | \mathcal{F}^S_t\right] dt + \sigma_{S} dN_t,
\end{eqnarray}
where $N$ is a $\left( \mathbb{P} ,\mathbf{F}^{S} \right)$ Wiener process.
\end{prop}
\begin{remarque}
In the filtering theory (see \cite{Lipster1} for details), the process $N$ is called the innovation process. To understand this name, note that:
\begin{eqnarray*}
dN_t  = \frac{1}{\sigma_{S}} \left(  \frac{dS_t}{S_t} - E\left[  \mu_t | \mathcal{F}^S_t\right] dt \right) .
\end{eqnarray*}
Then, $dN_t$ represents the difference between the current observation and what we expect knowing the past observations. 
\end{remarque}
\subsection{Optimal trend estimator} 
The system \hyperref[Model3]{(1)-(2)} corresponds to a Linear Gaussian Space State model (see \cite{ARMABook} for details). In this case, the Kalman filter gives the optimal estimator, which corresponds to the conditional expectation $E\left[  \mu_t | \mathcal{F}^S_t\right]$. Since $\left( \lambda,\sigma_{\mu},\sigma_{S}\right) \in  \mathbb{R}_+^{*}\times\mathbb{R}_+^{*}\times\mathbb{R}_+^{*}$, the model \hyperref[Model3]{(1)-(2)} is a controllable and observable time invariant system. In this case, it is well known  that the estimation error variance converges to an unique constant value (see \cite{KalmanBucy} for details). This corresponds to the steady-state Kalman filter. The following proposition (see \cite{Ahmed} for a proof) gives a first continuous representation of the steady-state Kalman filter:
\begin{prop}
The steady-state Kalman filter has a continuous time limit depending on the asset returns:
\begin{eqnarray}\label{ContinuousEstimate3}
d \widehat{\mu}_{t}=-\lambda \beta \widehat{\mu}_{t}dt+\lambda\left( \beta -1 \right) \frac{dS_{t}}{S_{t}}, 
\end{eqnarray}
where
\begin{eqnarray}\label{BetaDef3}
\beta= \left( 1+ \frac{\sigma_{\mu}^{2}}{ \lambda^{2} \sigma_{S}^{2}}  \right)^{\frac{1}{2}}.
\end{eqnarray}
\end{prop}
The steady-state Kalman filter can also be re-written as a corrected exponential average:
\begin{prop}\label{MMEproposition}
\begin{eqnarray}\label{MmeRepresentation}
\widehat{\mu}_{t}= m^*\tilde{\mu}_{t}^*,
\end{eqnarray}
where $m^*=\frac{\beta-1}{\beta}$ and $\tilde{\mu}^*$ is the exponential average given by:
\begin{eqnarray}\label{MMESDE}
d \tilde{\mu}_{t}^*=-\frac{1}{\tau^{*}} \tilde{\mu}_{t}^*dt+\frac{1}{\tau^{*}}\frac{dS_{t}}{S_{t}},
\end{eqnarray}
with an average duration $\tau^{*}=\frac{1}{\lambda\beta}$. 
\end{prop}
\section{Optimal strategy under parameters mis-specification}\label{optimaStrat}
In this section, we consider the optimal trading strategy under parameters mis-specification. For this portfolio, we first give the stochastic differential equation of the logarithmic return and we provide, 
in closed form, the asymptotic expectation of the logarithmic return.
\subsection{Context}
Consider the financial market defined in the first section with a risk free rate and without transaction costs.
Let $P$ be a self financing portfolio given by:
\begin{eqnarray*}
\frac{dP_{t}}{P_{t}}&=& \omega_t \frac{dS_{t}}{S_{t}},\\
P_0&=& x,
\end{eqnarray*}
where $\omega_t$ is the fraction of wealth invested in the risky asset (also named the control variable). The agent aims to maximize his expected logarithmic utility on an admissible domain $\mathcal{A}$ for the allocation process. In this section, we assume that the agent is not able to observe the trend $\mu$. Formally, $\mathcal{A}$ represents all the $\mathbf{F}^{S}$-progressive and measurable processes and the problem is:  
\begin{eqnarray*}
\omega^{*} = \arg\sup_{\omega \in \mathcal{A}} \mathbb{E} \left[\ln\left(P_{t} \right)| P_0=x \right].
\end{eqnarray*}
The solution of this problem is well known and easy to compute (see \cite{LaknerStrat} for example). Indeed, it has the following form:
\begin{eqnarray*}
\omega^{*}_t = \frac{E\left[  \mu_t | \mathcal{F}^S_t\right]}{\sigma_{S}^{2}}.
\end{eqnarray*}
In practice, the parameters are unknown and must be estimated. In \cite{Ahmed}, the authors assess the feasibility of forecasting trends modeled by an unobserved mean-reverting diffusion. They show that, due to a weak signal-to-noise ratio,  a bad calibration is very likely. Using Proposition \ref{MMEproposition}, the steady state Kalman filter is a corrected exponential moving average of past returns. Therefore, a mis-specification on the parameters $\left( \lambda,\sigma_{\mu}\right)$ is equivalent to a mis-specification on the factor $\frac{\beta-1}{\beta}$ and on the duration $\tau^{*}$. 

Suppose that an agent thinks that the optimal duration is $\tau$ and considers:
\begin{eqnarray}\label{MissSpecifiedFilter}
d \tilde{\mu}_{t}&=&-\frac{1}{\tau} \tilde{\mu}_{t}dt+\frac{1}{\tau}\frac{dS_{t}}{S_{t}},\\
\tilde{\mu}_{0}&=&0.
\end{eqnarray}
Using this estimator, the agent will invest following:
\begin{eqnarray}\label{MisSpecifiedSelfFinancing}
\frac{dP_{t}}{P_{t}}&=&m \frac{\tilde{\mu}_{t}}{\sigma_{S}^{2}}  \frac{dS_{t}}{S_{t}},\\
P_{0}&=& x,
\end{eqnarray}
where $m>0$. The following lemma gives the law of this filter $\tilde{\mu}$:
\begin{lemme}\label{LawMisSpecifiedFilter}
The exponential moving average of Equation (\ref{MissSpecifiedFilter}) is given by:
\begin{eqnarray}\label{SolutionMisFilter}
\tilde{\mu}_{t}=\frac{e^{-\frac{t}{\tau}}}{\tau}\left(\int_{0}^{t}e^{\frac{s}{\tau}}\mu_sds+\sigma_S\int_{0}^{t}e^{\frac{s}{\tau}}dW^S_s \right). 
\end{eqnarray}
Moreover, this filter is a centered Gaussian process, whose variance is:
\begin{eqnarray*}
&&\mathbb{V}\left[\tilde{\mu}_{t}\right] =  \frac{\sigma_S^2} {2\tau}\left(1-e^{\frac{-2t}{\tau}} \right) 
+\frac{\sigma_{\mu}^2} {\tau^2\lambda\left(\frac{1}{\tau}-\lambda \right) }\left( \tau\frac{e^{\frac{-2t}{\tau}}-1  }{2}\right. \\
&&\left.+\frac{1-e^{-t\left( \lambda+\frac{1}{\tau}\right) }}{\frac{1}{\tau}+\lambda}+ \frac{2e^{-t\left( \lambda+\frac{3}{\tau}\right)}-e^{-2t\left( \lambda+\frac{1}{\tau}\right)}-e^{\frac{-4t}{\tau}} }{\frac{1}{\tau}-\lambda}\right). 
\end{eqnarray*}
\end{lemme}
\begin{proof}
Applying It\^o's lemma to the function $f(\tilde{\mu}_t,t)=\tilde{\mu}_{t} e^{\frac{t}{\tau}}$ and using Equation (\ref{Model3}), it follows that:
\begin{eqnarray*}
df(\tilde{\mu}_t,t)=\frac{e^{\frac{t}{\tau}}}{\tau}\left(\mu_t dt + \sigma_S dW_t^S \right). 
\end{eqnarray*}
The integral of this stochastic differential equation from $0$ to $t$ gives Equation (\ref{SolutionMisFilter}). Therefore, $\tilde{\mu}$ is a Gaussian process. Its mean is null (because $\mu_0=0$). Since $\mu$ and $W^S$ are supposed to be independent, the variance of the process $\tilde{\mu}$ is equal to the sum of $\mathbb{V}\left[ \frac{e^{-\frac{t}{\tau}}}{\tau}\int_{0}^{t}e^{\frac{s}{\tau}}\mu_sds\right] $ and $\mathbb{V}\left[\frac{e^{-\frac{t}{\tau}}}{\tau}\sigma_S\int_{0}^{t}e^{\frac{s}{\tau}}dW^S_s\right] $. The first term is computed using:
\begin{eqnarray*}
\mathbb{V}\left[ \int_{0}^{t}e^{\frac{s}{\tau}}\mu_sds\right] =\int_{0}^{t}\int_{0}^{t} e^{\frac{s_1+s_2}{\tau}}\mathbb{E}\left[\mu_{s_1}\mu_{s_2}\right] ds_1ds_2.
\end{eqnarray*}
Since $\mu$ is a centered Ornstein Uhlenbeck, for all $s,t\geq 0$, we have:
\begin{eqnarray*}
\mathbb{E}\left[\mu_{s}\mu_{t}\right]=\mathbb{C}\text{ov}\left[\mu_s,\mu_t\right] =\frac{\sigma_{\mu}^{2}}{2\lambda} e^{-\lambda\left(s+t\right) }\left(e^{2\lambda s \wedge t}-1 \right).
\end{eqnarray*} 
Finally, the second term is computed using:
\begin{eqnarray*}
\mathbb{V}\left[\int_{0}^{t}e^{ks}dW^S_s\right]=\frac{1}{2k} \left( e^{2kt}-1\right),
\end{eqnarray*}
with $k>0$.
\end{proof}
\subsection{Portfolio dynamic}
The following proposition gives the stochastic differential equation of the mis-specified optimal portfolio:
\begin{prop}\label{MisPortfolioDynamic}
Equation (\ref{MisSpecifiedSelfFinancing}) leads to:
\begin{equation}
d \ln(P_{t}) =  \frac{m\tau}{2\sigma_S^2}d \tilde{\mu}_t^2+m\left(\frac{\tilde{\mu}_t ^2}{\sigma_S^2}\left( 1-\frac{m}{2}\right)-\frac{1}{2\tau}  \right) dt.
\end{equation}
\end{prop}
\begin{proof}
Equation  (\ref{MisSpecifiedSelfFinancing}) is equivalent to (by It\^o's lemma):
\begin{eqnarray*}
d \ln(P_{t})=  \frac{m\tilde{\mu}_t}{\sigma_{S}^{2}}  \frac{dS_{t}}{S_{t}} -   \frac{m^2\tilde{\mu}_t^{2}}{2\sigma_{S}^{2}}dt.
\end{eqnarray*}
Using Equation (\ref{MmeRepresentation}),
\begin{align*}
d \ln(P_{t}) =  \frac{m\tau}{\sigma_{S}^{2}} \tilde{\mu}_{t} d \tilde{\mu}_{t} +  \frac{m\tilde{\mu}_{t}^{2}}{\sigma_{S}^{2}} - \frac{1}{2} \frac{m^2\tilde{\mu}_{t}^{2}}{\sigma_{S}^{2}}dt,
\end{align*}
It\^o's lemma on Equation (\ref{MmeRepresentation}) gives:
\begin{eqnarray*}
d\tilde{\mu}_{t}^{2}=2\tilde{\mu}_{t}d\tilde{\mu}_{t}+\frac{\sigma_{S}^{2}}{\tau^2}dt.
\end{eqnarray*}
Using this equation, the dynamic of the logarithmic return follows.
\end{proof}
\begin{remarque}
Proposition \ref{MisPortfolioDynamic} shows that the returns of the optimal strategy can be broken down into two terms. The first one represents an option on the square of the realized returns (called Option profile). The second term is called the Trading Impact. These terms are introduced and discussed in \cite{Lyxor} for this strategy without considering a specific diffusion for the risky asset.
\end{remarque}
\subsection{Expected logarithmic return}
The following theorem gives the asymptotic expected logarithmic return of the mis-specified optimal strategy.
\begin{theo}\label{theoremExpMisspecified}
Consider the portfolio given by Equation (\ref{MisSpecifiedSelfFinancing}). In this case:
\begin{eqnarray}\label{AsymptoticExpMisSpecified}
\lim_{T \rightarrow \infty} \frac{\mathbb{E}\left[ \ln\left(  \frac{P_{T}}{P_{0}} \right)  \right] }{T}= m\frac{\tau\left(\beta^2-1 \right)\left(2-m \right)-m\left(\tau+\frac{1}{\lambda} \right)  }{4\tau\left(\tau+\frac{1}{\lambda} \right) },
\end{eqnarray}
where $\beta$ is given by Equation (\ref{BetaDef3}).
\end{theo}
\begin{proof}
Using Proposition \ref{MisPortfolioDynamic}, it follows that:
\begin{eqnarray*}
\mathbb{E}\left[ \ln\left(  \frac{P_{T}}{P_{0}} \right)\right] =  \frac{m\tau}{2\sigma_S^2} \mathbb{E}\left( \tilde{\mu}_T\right)^2+m\int_{0}^{T}\left(\frac{\mathbb{E} \left( \tilde{\mu}_t\right)^2\left( 2-m\right)}{2\sigma_S^2}-\frac{1}{2\tau}  \right) dt.
\end{eqnarray*}
Moreover, $\mathbb{E}\left[ \left( \tilde{\mu}_t\right)^2\right]$ is given by Lemma \ref{LawMisSpecifiedFilter}. Then, integrating the expression from $0$ to $T$ and tending $T$ to $\infty$, the result follows.
\end{proof}
The following result is a corollary of the previous theorem. It represents the asymptotic expected logarithmic return as a function of the signal-to-noise-ratio and of the trend mean reversion  speed $\lambda$.
\begin{cor}
Consider the portfolio given by Equation (\ref{MisSpecifiedSelfFinancing}). In this case:
\begin{eqnarray}\label{MisspecifiedAsymptoticExpSNR}
\lim_{T \rightarrow \infty} \frac{\mathbb{E}\left[ \ln\left(  \frac{P_{T}}{P_{0}} \right)  \right] }{T}= m\frac{2\tau\left(2-m \right)\text{SNR}-m\left(\lambda\tau+1 \right)  }{4\tau\left(\lambda\tau+1 \right) },
\end{eqnarray}
where $\text{SNR}$ is the signal-to-noise-ratio:
\begin{eqnarray}\label{SNR}
\text{SNR}=\frac{\sigma_{\mu}^{2}}{ 2\lambda \sigma_{S}^{2}}.
\end{eqnarray}
Moreover:
\begin{enumerate}
\item If $m < 2$, for a fixed parameter value $\lambda$,  this asymptotic expected logarithmic return is an increasing function of $\text{SNR}$.
\item For a fixed parameter value $\text{SNR}$, it is a decreasing function of  $\lambda$.
\end{enumerate} 
\end{cor}
\begin{proof}
Since $\beta=\sqrt{1+\frac{2\text{SNR}}{\lambda}}$, the use of this expression in Equation (\ref{AsymptoticExpMisSpecified}) gives the result.
\end{proof}
\begin{remarque}
Assume that the agent makes a good calibration and uses  $m^*=\frac{\beta-1}{\beta}$ and $\tau^{*}=\frac{1}{\lambda\beta}$. In this case, we obtain the result of \cite{Ahmed2}:
\begin{eqnarray}
\lim_{T \rightarrow \infty} \frac{\mathbb{E}\left[ \ln\left(  \frac{P_{T}}{P_{0}} \right)  \right] }{T}&=& \frac{1}{2}\left(\text{SNR}+\lambda-\sqrt{\lambda\left(\lambda+2\text{SNR} \right) }\right)\label{AsymptoticWellExpSNR},
\end{eqnarray} 
where $\text{SNR}$ is defined in Equation (\ref{SNR}).
\end{remarque}
The following proposition gives conditions on the trend parameters and on the duration $\tau$ that guarantee a positive asymptotic expected logarithmic return and the existence of an optimal duration.
\begin{prop}\label{OptimalDurationProp}
Consider the portfolio given by Equation (\ref{MisSpecifiedSelfFinancing}) and suppose that $m < 2$.
In this case, the asymptotic expected logarithmic return is positive if and only if:
\begin{enumerate}
\item $\frac{\text{SNR}}{\lambda}>\frac{2m}{2-m}$.
\item $\tau>\tau_{\text{min}}$, where:
\begin{eqnarray}\label{tauMin}
\tau_{\text{min}}=\frac{m}{2\left( 2-m\right)\text{SNR}-\lambda m }.
\end{eqnarray}
\end{enumerate}
Moreover, there exists an optimal duration $\tau_{\text{min}}<\tau_{\text{opt}}<\infty$ if and only if $\frac{\text{SNR}}{\lambda}>\frac{2m}{2-m}$ and:
\begin{eqnarray}\label{tauOpt}
\tau_{\text{opt}}=\frac{m + \sqrt{\left( 2-m\right)2m\frac{\text{SNR}}{\lambda}}  }{2\left( 2-m\right)\text{SNR}-\lambda m }.
\end{eqnarray}
\end{prop}
\begin{proof}
Using Equation (\ref{MisspecifiedAsymptoticExpSNR}), the first part of the proposition follows.
Since the asymptotic expected logarithmic return of the mis-specified strategy is positive after $\tau_{\text{min}}$ and tends to zero if $\tau$ tends to the infinity, there exists an optimal duration $\tau_{\text{opt}}$. This point is computed with setting to zero the derivative of Equation (\ref{MisspecifiedAsymptoticExpSNR}) with respect to the parameter $\tau$.
\end{proof}
\section{cross moving average strategy}\label{CrossAvSection}
In this section, we consider a cross moving average strategy based on geometric moving averages.  For this portfolio, we first give the stochastic differential equation of the logarithmic return and we provide, 
in closed form, the asymptotic expectation of the logarithmic return.  
\subsection{Context}
Consider the financial market defined in the first section with a risk free rate and without transaction costs. Let $G\left(t,L\right)$ be the geometric moving average at time $t$ of the stock prices on a window $L$:
\begin{eqnarray}\label{GeomAvrg}
G\left(t,L\right)=\exp\left( \frac{1}{L}\int_{t-L}^{t}\log\left(S_u \right)du \right). 
\end{eqnarray}
Let $Q$ be a self financing portfolio given by:
\begin{eqnarray}\label{TecAnalysisPortfolio}
\frac{dQ_{t}}{Q_{t}}&=& \theta_t \frac{dS_{t}}{S_{t}},\\
Q_0&=& x,
\end{eqnarray}
where $\theta_t$ is the fraction of wealth invested by the agent in the risky asset:
\begin{eqnarray*}
\theta_t =\gamma+\alpha \indicator{G\left(t,L_1\right)>G\left(t,L_2\right)}
\end{eqnarray*}
with $\gamma,\alpha \in  \mathbb{R}$ and $0<L_1<L_2<t$. This trading strategy is a combination of a fixed strategy and a pure cross moving average strategy. 
\subsection{Portfolio dynamic}
The following proposition gives the stochastic differential equation of the cross moving average portfolio.
\begin{prop}\label{TechPortfolioDynamic}
Equation (\ref{TecAnalysisPortfolio}) leads to:
\begin{eqnarray*}
d \ln(Q_{t}) &=&\left(  \left(\gamma+\alpha \indicator{G\left(t,L_1\right)>G\left(t,L_2\right)} \right)\mu_t -\frac{\gamma^2\sigma_S^2}{2}\right. \\
&&\left. -\frac{\left( \alpha^2 +2\alpha\gamma\right) \sigma_S^2}{2} \indicator{G\left(t,L_1\right)>G\left(t,L_2\right)}\right) dt\\
&&+ \left(\gamma+\alpha \indicator{G\left(t,L_1\right)>G\left(t,L_2\right)} \right)\sigma_S dW_{t}^{S}.
\end{eqnarray*}
\end{prop}
\begin{proof}
Applying It\^o's lemma to the process $\ln(Q)$ and using
\begin{eqnarray*}
\indicator{G\left(t,L_1\right)>G\left(t,L_2\right)}^2=\indicator{G\left(t,L_1\right)>G\left(t,L_2\right)},
\end{eqnarray*}
Proposition \ref{TechPortfolioDynamic} follows.
\end{proof}
\subsection{Expected logarithmic return}
The following theorem gives the asymptotic expected logarithmic return of the cross moving average portfolio.
\begin{theo}\label{theoremTechAnal}
Consider the portfolio given by Equation (\ref{TecAnalysisPortfolio}). In this case:
\begin{eqnarray*}\label{AsymptoticTheoremTechAnal}
\lim_{T \rightarrow \infty} \frac{\mathbb{E}\left[ \ln\left(  \frac{Q_{T}}{Q_{0}} \right)  \right] }{T}=-\frac{\gamma^2\sigma_S^2}{2}-\frac{\left( \alpha^2 +2\alpha\gamma\right) \sigma_S^2}{2} \Phi\left(\frac{m_{\left( L_1,L_2,\sigma_S\right)}}{\sqrt{s_{\left( L_1,L_2,\lambda,\sigma_\mu,\sigma_S\right)}}}\right)\\
+ \frac{\alpha\sigma_\mu^2 \left( L_2\left(1-e^{-\lambda L_1}\right)-L_1\left(1-e^{-\lambda L_2} \right) \right)}{2 \lambda^3L_1L_2 \sqrt{s_{\left( L_1,L_2,\lambda,\sigma_\mu,\sigma_S\right)} } }   \Phi'\left(-\frac{m_{\left( L_1,L_2,\sigma_S\right)}}{\sqrt{s_{\left( L_1,L_2,\lambda,\sigma_\mu,\sigma_S\right)}}}\right),
\end{eqnarray*}
where $\Phi$ is the cumulative distribution function of the standard normal variable and:
\begin{eqnarray*}
m_{\left( L_1,L_2,\sigma_S\right)}&=&\frac{-\sigma_S^2}{4}\left( L_2-L_1\right)  ,\\
s_{\left( L_1,L_2,\lambda,\sigma_\mu,\sigma_S\right)}&=&\left(\frac{\sigma_\mu^2}{\lambda^2}+\sigma_S^2 \right)\frac{\left(L_2-L_1 \right)^2 }{3L_2} -\frac{\sigma_\mu^2}{\lambda^4}\left( \frac{1}{L_1}-\frac{1}{L_2}\right) \\
&&+\frac{\sigma_\mu^2}{\lambda^5}\left[ \frac{1}{L_1^2}\left(1-e^{-\lambda L_1} \right) +\frac{1}{L_2^2}\left(1-e^{-\lambda L_2} \right)\right. \\
&&\left. -\frac{1}{L_1L_2}\left(1-e^{-\lambda L_1} \right)\left(1-e^{-\lambda L_2} \right)\right. \\
&&\left.-\frac{1}{L_1L_2}\left(e^{-\lambda\left( L_2-L_1\right) }-e^{-\lambda\left( L_2+L_1\right) } \right) \right]. 
\end{eqnarray*}
\end{theo}
\begin{proof}
Since the processes $\mu$ and $W^{S}$ are centered, Proposition \ref{TechPortfolioDynamic} implies that:
\begin{eqnarray*}
\mathbb{E}\left[ \ln\left(  \frac{Q_{T}}{Q_{0}} \right)  \right]&=&\frac{-\gamma^2\sigma_S^2}{2}\left( T-L_2\right) \\
&&+\alpha \int_{L_2}^{T} \mathbb{E}\left[\mu_t\indicator{G\left(t,L_1\right)>G\left(t,L_2\right)}\right] dt\\ &&-\frac{\left( \alpha^2 +2\alpha\gamma\right) \sigma_S^2}{2}\int_{L_2}^{T}\mathbb{E}\left[\indicator{G\left(t,L_1\right)>G\left(t,L_2\right)}\right] dt,
\end{eqnarray*}
where $T>L_2$. Let $t>L_2$ and consider the following process:
\begin{eqnarray}\label{Xprocess}
X_t&=&m_1\left( t\right) -m_2\left( t\right) ,
\end{eqnarray}
where $\forall i \in \left\lbrace 1,2\right\rbrace$:
\begin{eqnarray*}
m_i\left( t\right)&=& \frac{1}{L_i}\int_{t-L_i}^{t}\log\left(S_u \right)du.
\end{eqnarray*}
Then $X$ is a Gaussian process. Based on Lemma 2 in \cite{Zhou}, $\forall t>L_2$:
\begin{eqnarray}
\left\lbrace G\left(t,L_1\right)>G\left(t,L_2\right)\right\rbrace &\Leftrightarrow& \left\lbrace X_t >0\right\rbrace,\\
\mathbb{E}\left[\indicator{G\left(t,L_1\right)>G\left(t,L_2\right)}\right]&=&\Phi\left(\frac{\mathbb{E}\left[X_t\right] }{\sqrt{\mathbb{V}ar\left[X_t\right]}}\right) ,\\
\mathbb{E}\left[\mu_t\indicator{G\left(t,L_1\right)>G\left(t,L_2\right)}\right]&=&\frac{\mathbb{C}ov\left[X_t,\mu_t\right]}{\sqrt{\mathbb{V}ar\left[X_t\right]}}\Phi'\left(-\frac{\mathbb{E}\left[X_t\right] }{\sqrt{\mathbb{V}ar\left[X_t\right]}}\right)
\end{eqnarray}
The following lemma gives the mean, the asymptotic variance of the process $X$ and the covariance function between the processes $X$ and $\mu$.
\begin{lemme}\label{ProcessX}
Consider the process $X$ defined in Equation (\ref{Xprocess}). In this case, $\forall t>L_2$:
\begin{eqnarray}
\mathbb{E}\left[X_t\right]&=&\frac{-\sigma_S^2}{4}\left( L_2-L_1\right) \label{EqMoyenne} ,\\
\lim_{t \rightarrow \infty} \mathbb{V}ar\left[X_t\right]&=&s_{\left( L_1,L_2,\lambda,\sigma_\mu,\sigma_S\right)},\label{EqVar}\\
\mathbb{C}ov\left[X_t,\mu_t\right]&=& g\left(t,L_1 \right) -g\left(t,L_2 \right),\label{Cov}
\end{eqnarray}
where $s_{\left( L_1,L_2,\lambda,\sigma_\mu,\sigma_S\right)}$ is defined in Theorem \ref{theoremTechAnal} and
\begin{eqnarray}\label{FonctionCov}
g\left(t,L \right)=\frac{-\sigma_\mu^2 e^{-\lambda t}}{\lambda^2 L}\left(\lambda L + \sinh\left(\lambda\left(t-L \right)  \right)  - \sinh\left(\lambda t\right) \right). 
\end{eqnarray}
\end{lemme}
\begin{proof}[Proof of Lemma \ref{ProcessX}]
Since: 
\begin{eqnarray*}
\mathbb{E}\left[m_i\left( t\right)\right]=\frac{-\sigma_S^2}{4}\left( 2t-L_i\right),
\end{eqnarray*}
Equation (\ref{EqMoyenne}) follows. Moreover:
\begin{eqnarray*}
\mathbb{C}ov\left[m_1\left( t\right),m_2\left( t\right)\right]=\frac{1}{L_1L_2}\int_{t-L_1}^{t}\int_{t-L_2}^{t}\mathbb{C}ov\left[\ln S_u  ,\ln S_v  \right]du dv,
\end{eqnarray*}
Since
\begin{eqnarray*}
\mathbb{C}ov\left[\ln S_u  ,\ln S_v  \right]=\int_{0}^{u}\int_{0}^{v}\mathbb{C}ov\left[\mu_s  ,\mu_t  \right]ds dt+\sigma_S^2 \min\left( u,v\right),
\end{eqnarray*}
and the drift $\mu$ is an Ornstein Uhlenbeck process:
\begin{eqnarray*}
\mathbb{C}ov\left[\mu_s  ,\mu_t  \right]=\frac{\sigma_\mu^2 e^{-\lambda \left(s+t \right) }}{2\lambda}\left(e^{2\lambda \min\left( s,t\right) } -1\right). 
\end{eqnarray*}
Then
\begin{eqnarray*}
\mathbb{C}ov\left[\ln S_u  ,\ln S_v  \right]&=&\left( \sigma_S^2+\frac{\sigma_\mu^2}{\lambda^2}\right)  \min\left( u,v\right)\\
&&+\frac{\sigma_\mu^2}{2\lambda^3}\left(2e^{-\lambda u}+2e^{-\lambda v}-e^{-\lambda \left|v-u \right| }-e^{-\lambda \left(v+u \right) }-1 \right). 
\end{eqnarray*}
Using 
\begin{eqnarray*}
\mathbb{V}ar\left[X_t\right]=\mathbb{V}ar\left[m_1\left( t\right)\right] +\mathbb{V}ar\left[m_2\left( t\right)\right] -2\mathbb{C}ov\left[m_1\left( t\right),m_2\left( t\right)\right]
\end{eqnarray*}
and tending t to $\infty$ Equation (\ref{EqVar}) follows. Since the processes $W^S$ and $\mu$ are supposed to be independent, there holds:
\begin{eqnarray*}
\mathbb{C}ov\left[X_t,\mu_t\right]=\mathbb{C}ov\left[m_1\left( t\right),\mu_t\right]-\mathbb{C}ov\left[m_2\left( t\right),\mu_t\right].
\end{eqnarray*}
Moreover
\begin{eqnarray*}
\mathbb{C}ov\left[m_i\left( t\right),\mu_t\right]=\frac{1}{L_i}\int_{t-L_i}^{t}\mathbb{C}ov\left[\ln S_u,\mu_t\right]du,
\end{eqnarray*}
and 
\begin{eqnarray*}
\mathbb{C}ov\left[\ln S_u,\mu_t\right]=\int_{0}^{u}\mathbb{C}ov\left[\mu_s,\mu_t\right] ds,
\end{eqnarray*}
then 
\begin{eqnarray*}
\mathbb{C}ov\left[m_i\left( t\right),\mu_t\right]=g\left(t,L_i \right),
\end{eqnarray*}
where the function $g$ is defined in Equation (\ref{FonctionCov}). Equation (\ref{Cov}) follows
\end{proof}
The use of Lemma \ref{ProcessX} gives:
\begin{eqnarray*}
\mathbb{E}\left[ \ln\left(  \frac{Q_{T}}{Q_{0}} \right)  \right]&=&\frac{-\gamma^2\sigma_S^2}{2}\left( T-L_2\right) \\
&&+\alpha\Phi'\left(-\frac{m_{\left( L_1,L_2,\sigma_S\right)}}{\sqrt{s_{\left( L_1,L_2,\lambda,\sigma_\mu,\sigma_S\right)}}}\right) \int_{L_2}^{T}\frac{\mathbb{C}ov\left[X_t,\mu_t\right]}{\sqrt{\mathbb{V}ar\left[X_t\right]}} dt\\ 
&&-\frac{\left( \alpha^2 +2\alpha\gamma\right) \sigma_S^2}{2} \left( T-L_2\right)   \Phi\left(\frac{m_{\left( L_1,L_2,\sigma_S\right)}}{\sqrt{s_{\left( L_1,L_2,\lambda,\sigma_\mu,\sigma_S\right)}}}\right).
\end{eqnarray*}
Moreover, a direct calculus shows that:
\begin{eqnarray*}
\lim_{T \rightarrow \infty} \frac{\int_{L_2}^{T}\frac{\mathbb{C}ov\left[X_t,\mu_t\right]}{\sqrt{\mathbb{V}ar\left[X_t\right]}}dt}{T}=\frac{\sigma_\mu^2 \left(L_2\left(1-e^{-\lambda L_1} \right)-L_1\left(1-e^{-\lambda L_2} \right)  \right) }{2\lambda^3L_1L_2\sqrt{s_{\left( L_1,L_2,\lambda,\sigma_\mu,\sigma_S\right)}}}, 
\end{eqnarray*}
the result of Theorem \ref{theoremTechAnal} follows.
\end{proof}
\subsection{Strategy with one moving average}
Suppose that $L_1=0$ and $L_2=L$. In this case, the fraction of wealth invested by the agent in the risky asset becomes:
\begin{eqnarray*}
\theta_t^1 =\gamma+\alpha \indicator{S_t>G\left(t,L\right)},
\end{eqnarray*}
where $G$ is the geometric moving average defined in Equation (\ref{GeomAvrg}) and the self financing portfolio $Q^1$ becomes:
\begin{eqnarray}\label{TecAnalysisPortfolio1}
\frac{dQ_{t}^1}{Q_{t}^1}&=& \theta_t^1 \frac{dS_{t}}{S_{t}},\\
Q_0^1&=& x,
\end{eqnarray}
This particular case corresponds to the allocation introduced in \cite{Zhou} when we assume that the two Brownian motions $W^{S}$ and $W^{\mu}$ are uncorrelated and that the trend is mean reverted around $0$. Given this framework, we can provide the asymptotic expected logarithmic return of this trading strategy (which has already been found in \cite{Zhou}):
\begin{theo}\label{theoremTechAnal1}
Consider the portfolio given by Equation (\ref{TecAnalysisPortfolio1}). In this case:
\begin{eqnarray*}\label{AsymptoticTheoremTechAnal1}
\lim_{T \rightarrow \infty} \frac{\mathbb{E}\left[ \ln\left(  \frac{Q_{T}^1}{Q_{0}^1} \right)  \right] }{T}=-\frac{\gamma^2\sigma_S^2}{2}-\frac{\left( \alpha^2 +2\alpha\gamma\right) \sigma_S^2}{2} \Phi\left(\frac{m^1_{\left( L,\sigma_S\right)}}{\sqrt{s^1_{\left( L,\lambda,\sigma_\mu,\sigma_S\right)}}}\right)\\
+ \frac{\alpha\sigma_\mu^2 \frac{1-\left( 1-e^{-\lambda L}\right)}{\lambda L} }{2 \lambda^2 \sqrt{s^1_{\left(L,\lambda,\sigma_\mu,\sigma_S\right)} } }   \Phi'\left(-\frac{m^1_{\left( L,\sigma_S\right)}}{\sqrt{s^1_{\left( L,\lambda,\sigma_\mu,\sigma_S\right)}}}\right),
\end{eqnarray*}
where $\Phi$ is the cumulative distribution function of the standard normal variable and:
\begin{eqnarray*}
m^1_{\left( L,\sigma_S\right)}&=& m_{\left( 0,L,\sigma_S\right)}\\
&=&\frac{-\sigma_S^2}{4}L  ,\\
s^1_{\left(L,\lambda,\sigma_\mu,\sigma_S\right)}&=& s_{\left( 0,L,\lambda,\sigma_\mu,\sigma_S\right)}\\
&=&\frac{\left(\frac{\sigma_\mu^2}{\lambda^2}+\sigma_S^2 \right) L}{3} - \frac{\sigma_\mu^2}{2\lambda^3}\left(1-\frac{2\left( 1-e^{-\lambda L}\left(1+\lambda L \right) \right) }{\lambda^2L^2} \right), 
\end{eqnarray*}
and the functions $s$ and $m$ are introduced in Theorem \ref{theoremTechAnal}. 
\end{theo}
\begin{proof}
This result is a consequence of  Theorem \ref{theoremTechAnal}. Indeed, tending $L_1$ to $0$ and using $L_2=L$, the result follows.
\end{proof}

\section{Simulations}
In this section, numerical simulations and empirical tests based on real data are performed. The aim of these tests is to compare the robustness of the optimal strategy under parameters  mis-specification and of an investment using cross moving averages. First, the best durations of the Kalman filter and of the optimal strategy under parameters  mis-specification are illustrated over several trend regimes. We then consider the asymptotic expected logarithmic returns of the cross moving average strategy (see Section \ref{CrossAvSection}) with $\left(L_1,L_2 \right)=\left(5 \text{ days},252 \text{ days}\right) $  and of the optimal strategy with a duration $\tau=252$ days. Using this configuration, we study the stability of the performances of these strategies over several theoretical regimes. We also confirm our results under Heston's stochastic volatility model with Monte Carlo simulation.
Finally, backtests of these two strategies on real data confirm our theoretical expectations.
\subsection{Optimal durations}
In this subsection, we consider the model \hyperref[Model3]{(1)-(2)}.
\subsubsection{Well-specified Kalman filter}
In these simulations, we consider a signal-to-noise ratio inferior to 1. This assumption corresponds to a trend standard deviation inferior to the volatility of the risky asset. Using $\tau^{*}=\frac{1}{\lambda\beta}$ and $\beta=\sqrt{1+\frac{2\text{SNR}}{\lambda}}$, The figures \hyperref[Figu31]{1} and \hyperref[Figu32]{2} represent the optimal Kalman filter duration $\tau^{*}$ as a function of the trend mean reversion speed $\lambda$ and of the signal-to-noise ratio. This duration is a decreasing function of these parameters. Indeed, if the variation of the trend process is  low and if the measurement noise is high compared to the trend standard deviation, the window of filtering must be long. Moreover, we observe that for a trend mean reversion speed inferior to 1 (which corresponds to a slow trend process), the duration $\tau^{*}$ is superior to 0.5 years and can reach 10 years.  If the trend mean reversion speed is superior to 1, this duration is inferior to 1 year.
\begin{figure}[H]
\begin{center}
   \includegraphics[totalheight=5cm]{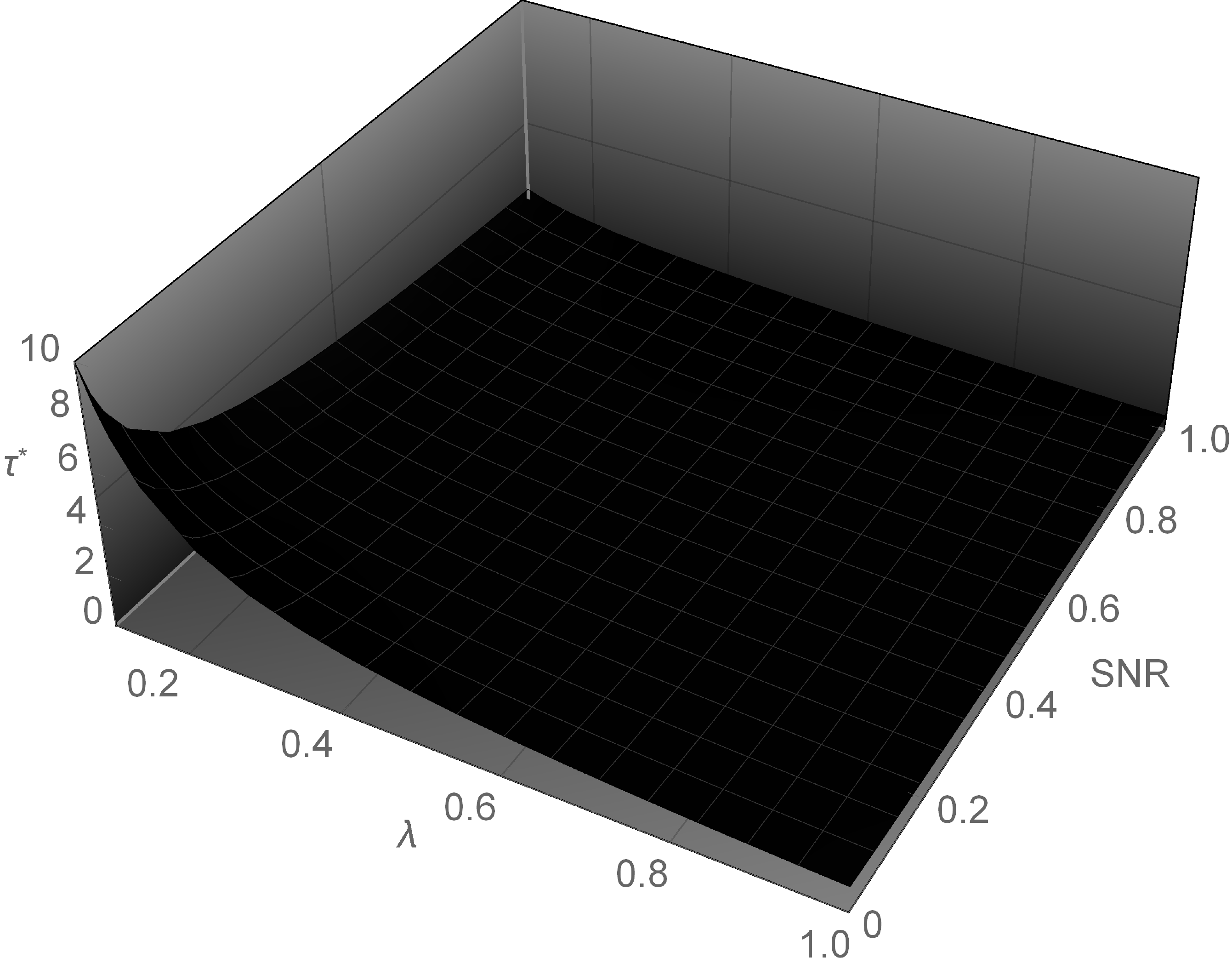}   
  \caption{Optimal duration (years) of the Kalman filter with $\lambda \in \left[ 0.1,1\right]$ }
\end{center}\label{Figu31}
\end{figure}
\begin{figure}[H]
\begin{center}
   \includegraphics[totalheight=4cm]{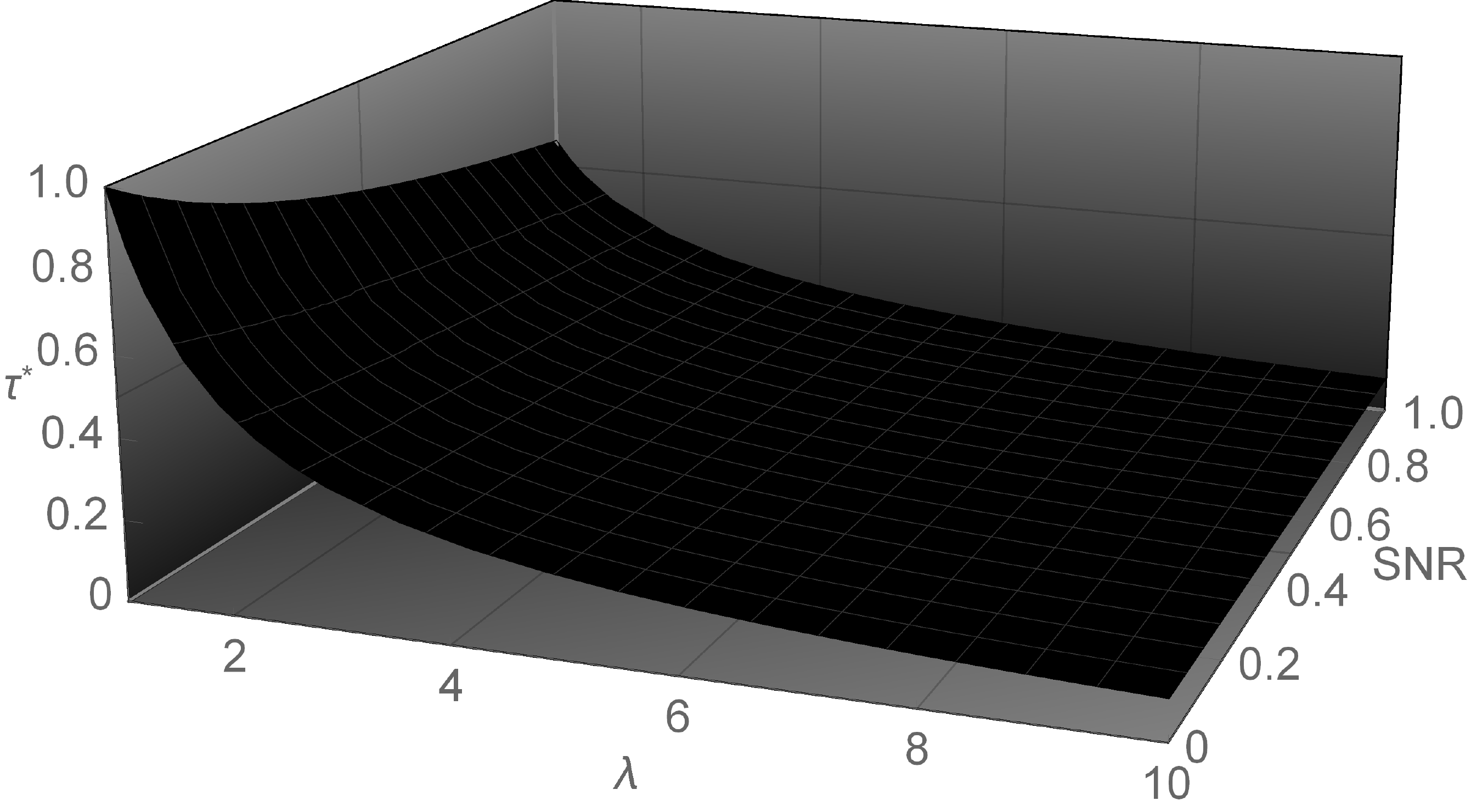}   
  \caption{Optimal duration (years) of the Kalman filter with $\lambda \in \left[ 1,10\right]$ }
\end{center}\label{Figu32}
\end{figure}
\subsubsection{Best filtering window for the optimal strategy under parameters  mis-specification}
Under parameters  mis-specification, we can also define an optimal duration using the strategy introduced in Section \ref{optimaStrat} and Proposition \ref{OptimalDurationProp}. This duration is the one maximizing the asymptotic expected logarithmic return of the optimal strategy under parameters mis-specification. This optimal window exists if and only if $\frac{\text{SNR}}{\lambda}>\frac{2m}{2-m}$. We assume that $m=1$. Then, the condition becomes $\frac{\text{SNR}}{\lambda}>2$.  The figures \hyperref[Figu33]{3} and \hyperref[Figu34]{4} represent this duration $\tau_{\text{opt}}\left(m=1\right)$ as a function of the trend mean reversion speed $\lambda$ with respectively SNR$=1$ and SNR$=0.5$.
This duration  has a similar behaviour than the optimal Kalman filter duration, except when  the trend mean reversion speed $\lambda$ tends to  $\frac{\text{SNR}}{2}$. Indeed, if $\lambda=\frac{\text{SNR}}{2}$, the condition $\frac{\text{SNR}}{\lambda}>2$ is not satisfied  and the optimal duration becomes infinite.
\begin{figure}[H]
\begin{center}
   \includegraphics[totalheight=7cm]{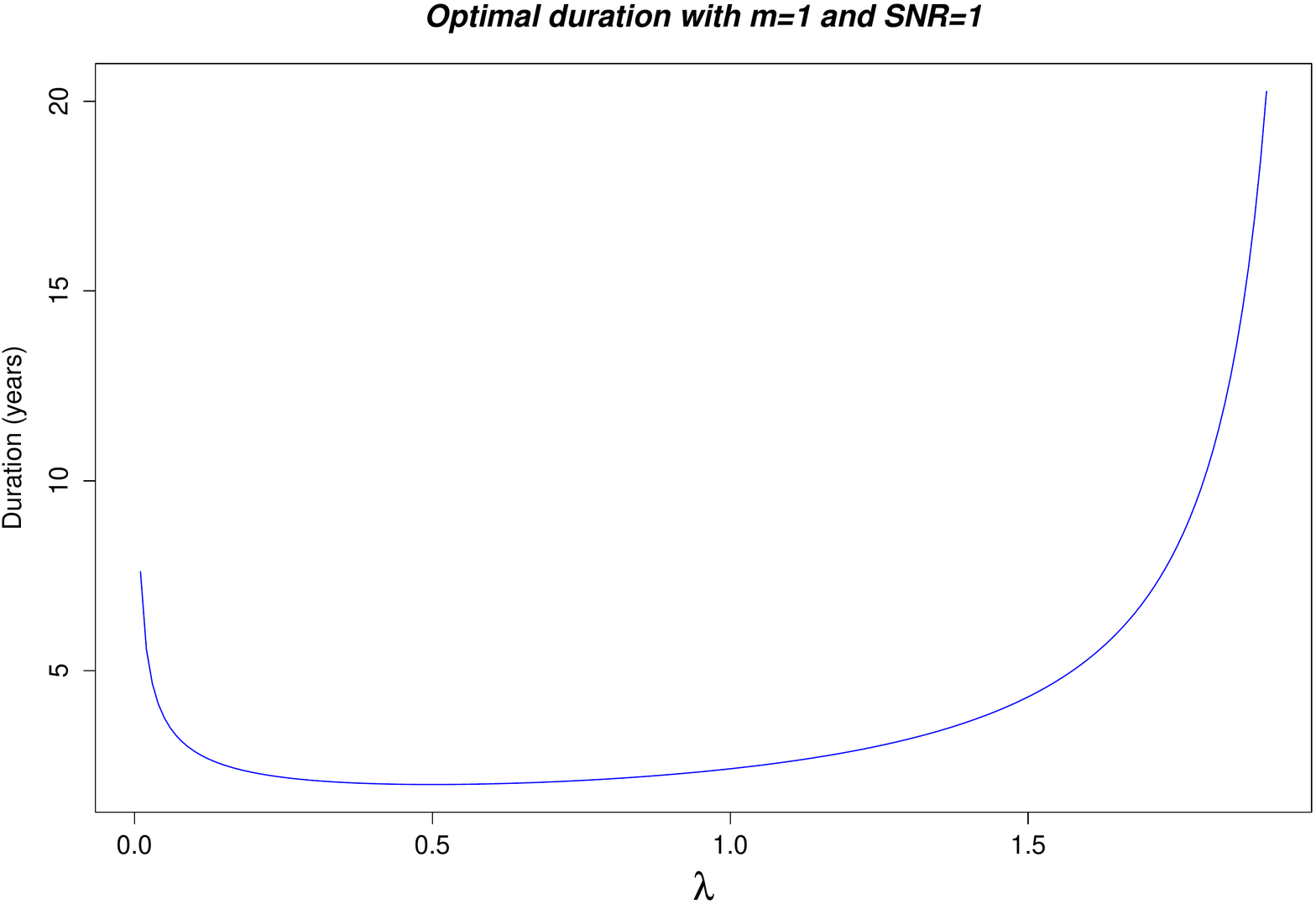}   
  \caption{Optimal duration (years) of the mis-specified filter with $m=1$ and SNR$=1$}
\end{center}\label{Figu33}
\end{figure}

\begin{figure}[H]
\begin{center}
   \includegraphics[totalheight=7cm]{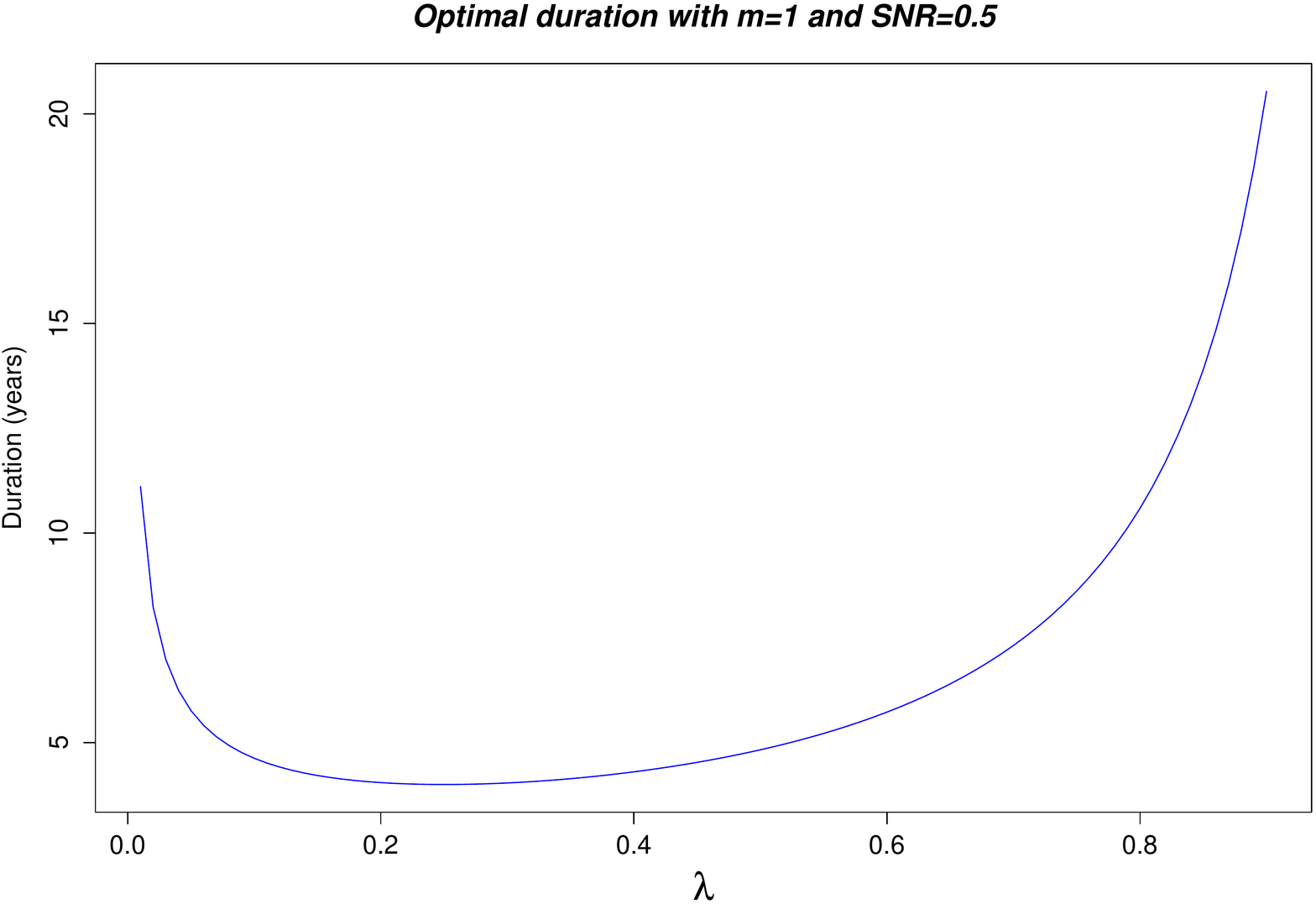}   
  \caption{Optimal duration (years) of the mis-specified filter with $m=1$ and SNR$=0.5$}
\end{center}\label{Figu34}
\end{figure}
\subsection{Robustness of the optimal strategy and of the cross moving average strategy}
\subsubsection{Stability of the performances over several theoretical regimes under constant spot volatility}
In this subsection, we consider the model \hyperref[Model3]{(1)-(2)}. Moreover, we assume that a year contains 252 days and that the risky asset volatility is equal to $\sigma_S=30\%$. We consider two trading strategies. The first one is the optimal strategy (introduced in section \ref{optimaStrat}) with a duration $\tau=252$ days ($=1$ year) and a leverage $m=1$. The second strategy is the cross moving average strategy (introduced in section \ref{CrossAvSection}) with $\left(L_1,L_2 \right)=\left(5 \text{ days},252 \text{ days}\right) $  and the following allocation:
\begin{eqnarray*}
\theta_t =-1+2\indicator{G\left(t,L_1\right)>G\left(t,L_2\right)},
\end{eqnarray*}
where $G$ is the geometric moving average defined in Equation (\ref{GeomAvrg}). Then, if the short geometric average is superior (respectively inferior) to the long geometric average, we buy (respectively sell) the risky asset. In order to compare the performance stability of these two strategies, we use the asymptotic expected logarithmic returns found in Theorems \ref{theoremExpMisspecified} and $\ref{theoremTechAnal}$.
The figures \hyperref[Figu35]{5}, \hyperref[Figu36]{6}, \hyperref[Figu37]{7} and \hyperref[Figu38]{8} represent the performances of these strategies after 100 years as a function of the trend volatility $\sigma_\mu$ respectively with $\lambda=1,2,3$ and 4. Even if the optimal strategy can provide a better performance (for example with $\lambda=1$ and $\sigma_\mu=90\%$ ), it can also provide higher losses than the cross average strategy (for example with $\lambda=4$ and $\sigma_\mu=10\%$). We can conclude with these tests that the theoretical performance of this cross average strategy is more robust than the theoretical performance of this optimal strategy. 
\begin{figure}[H]
\begin{center}
   \includegraphics[totalheight=7cm]{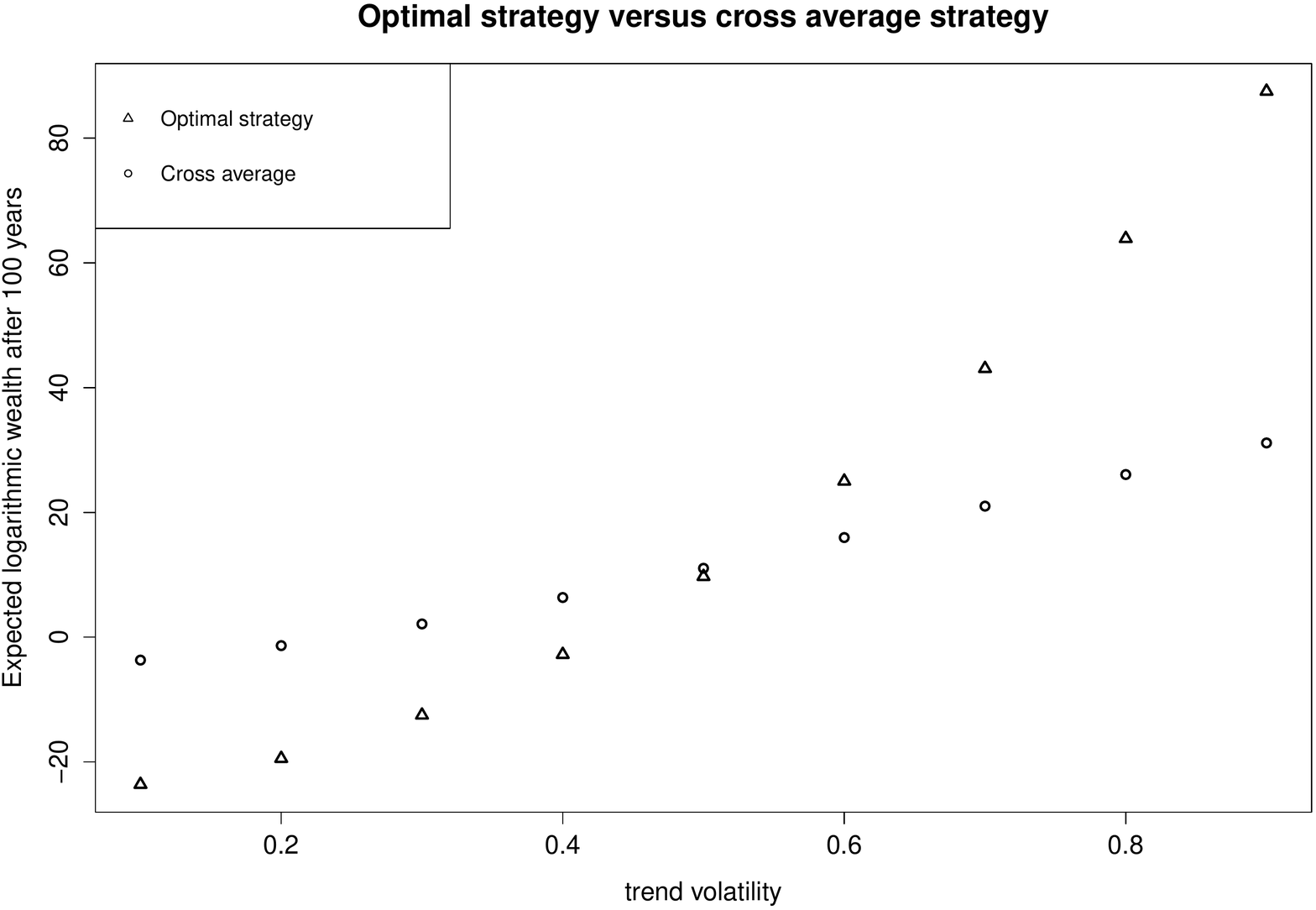}   
  \caption{The expected logarithmic returns of the optimal strategy (with $\tau=252$ days) and of the cross average strategy ($L_1=5$ days and $L_1=252$ days) as functions of $\sigma_\mu$ with $\lambda=1$, $\sigma_S=30\%$ and $T=100$ years}
\end{center}\label{Figu35}
\end{figure}
\begin{figure}[H]
\begin{center}
   \includegraphics[totalheight=7cm]{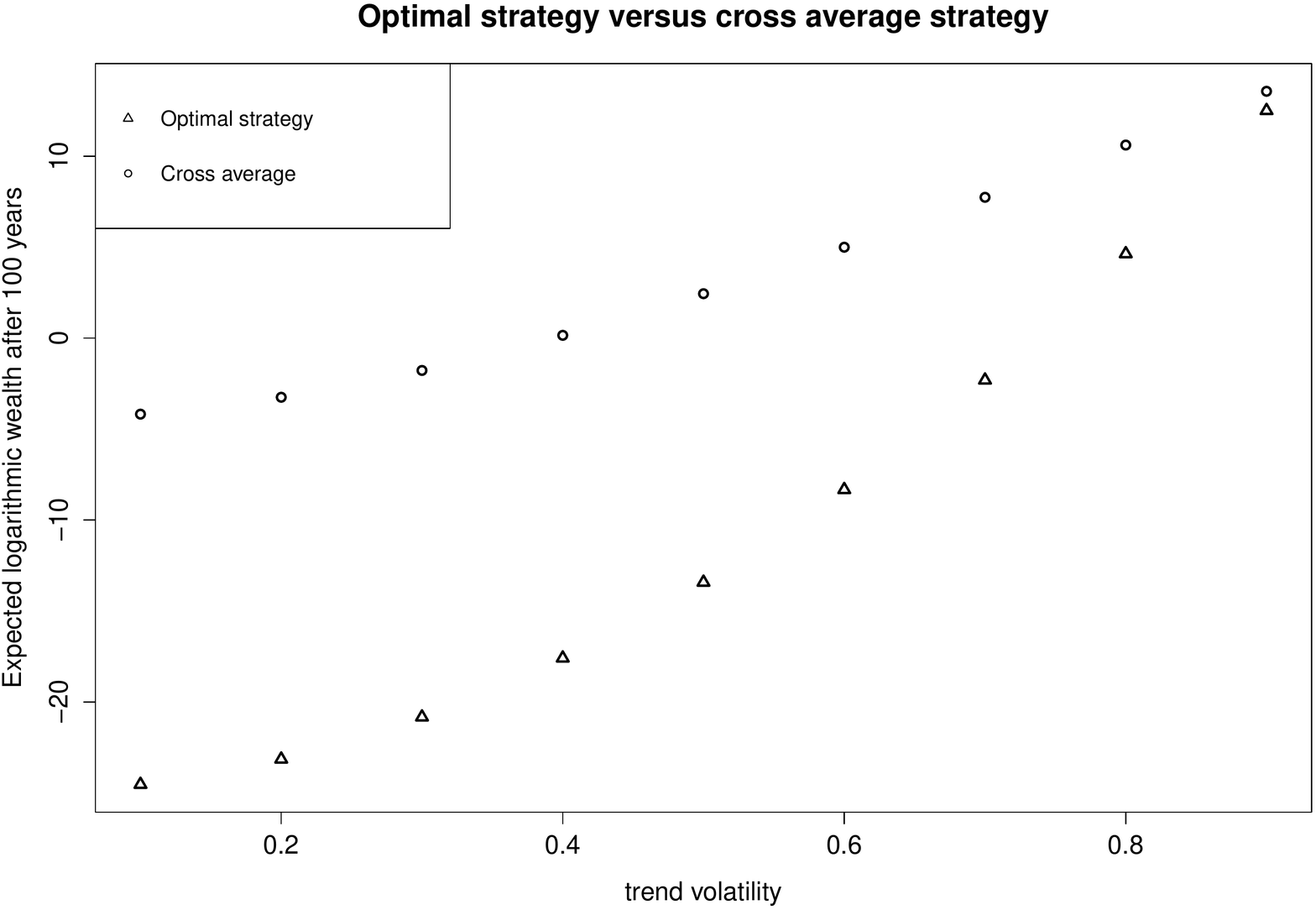}   
  \caption{The expected logarithmic returns of the optimal strategy (with $\tau=252$ days) and of the cross average strategy ($L_1=5$ days and $L_1=252$ days) as functions of $\sigma_\mu$ with $\lambda=2$, $\sigma_S=30\%$ and $T=100$ years}
\end{center}\label{Figu36}
\end{figure}
\begin{figure}[H]
\begin{center}
   \includegraphics[totalheight=7cm]{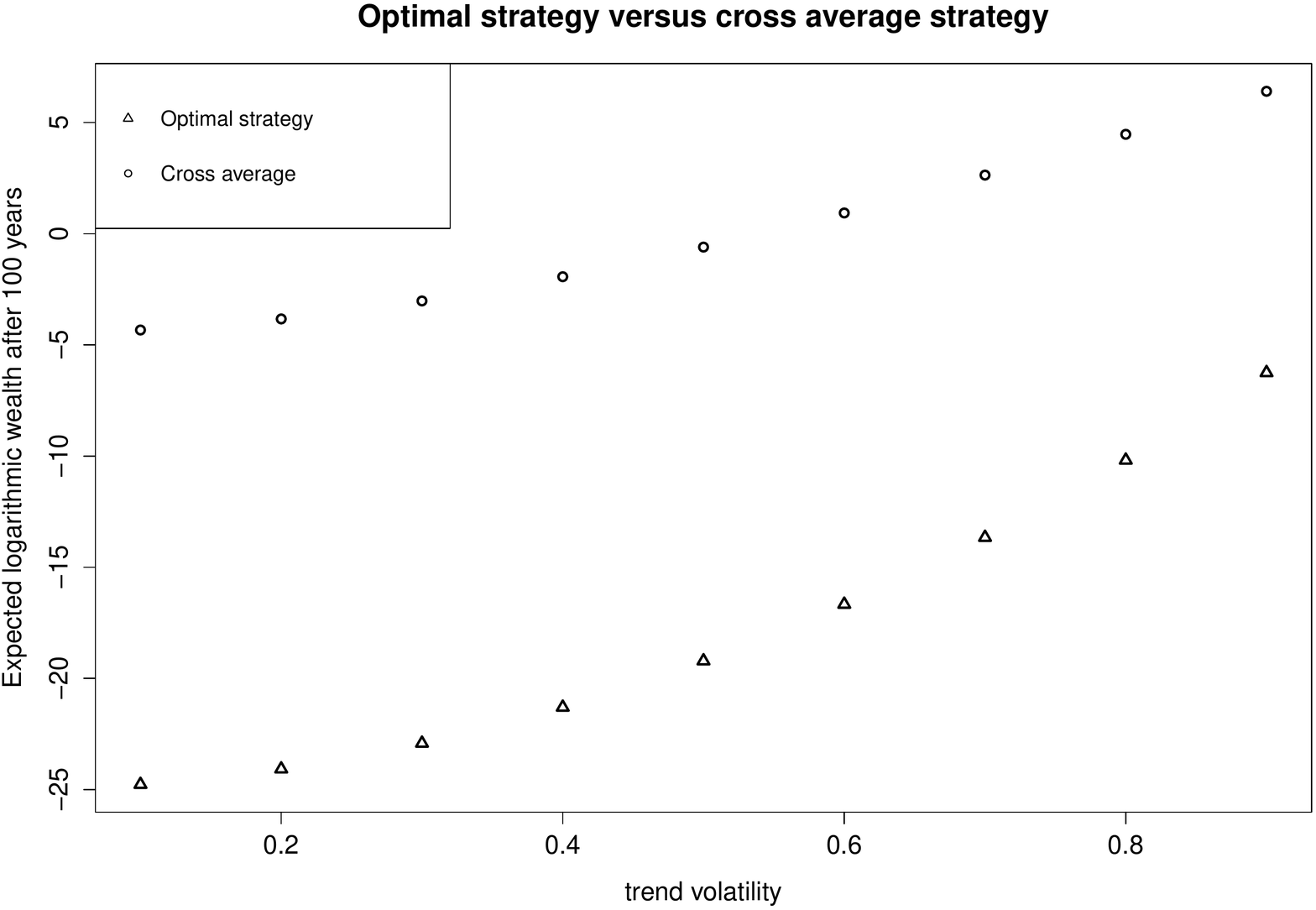}   
  \caption{The expected logarithmic returns of the optimal strategy (with $\tau=252$ days) and of the cross average strategy ($L_1=5$ days and $L_1=252$ days) as functions of $\sigma_\mu$ with $\lambda=3$, $\sigma_S=30\%$ and $T=100$ years}
\end{center}\label{Figu37}
\end{figure}
\begin{figure}[H]
\begin{center}
   \includegraphics[totalheight=7cm]{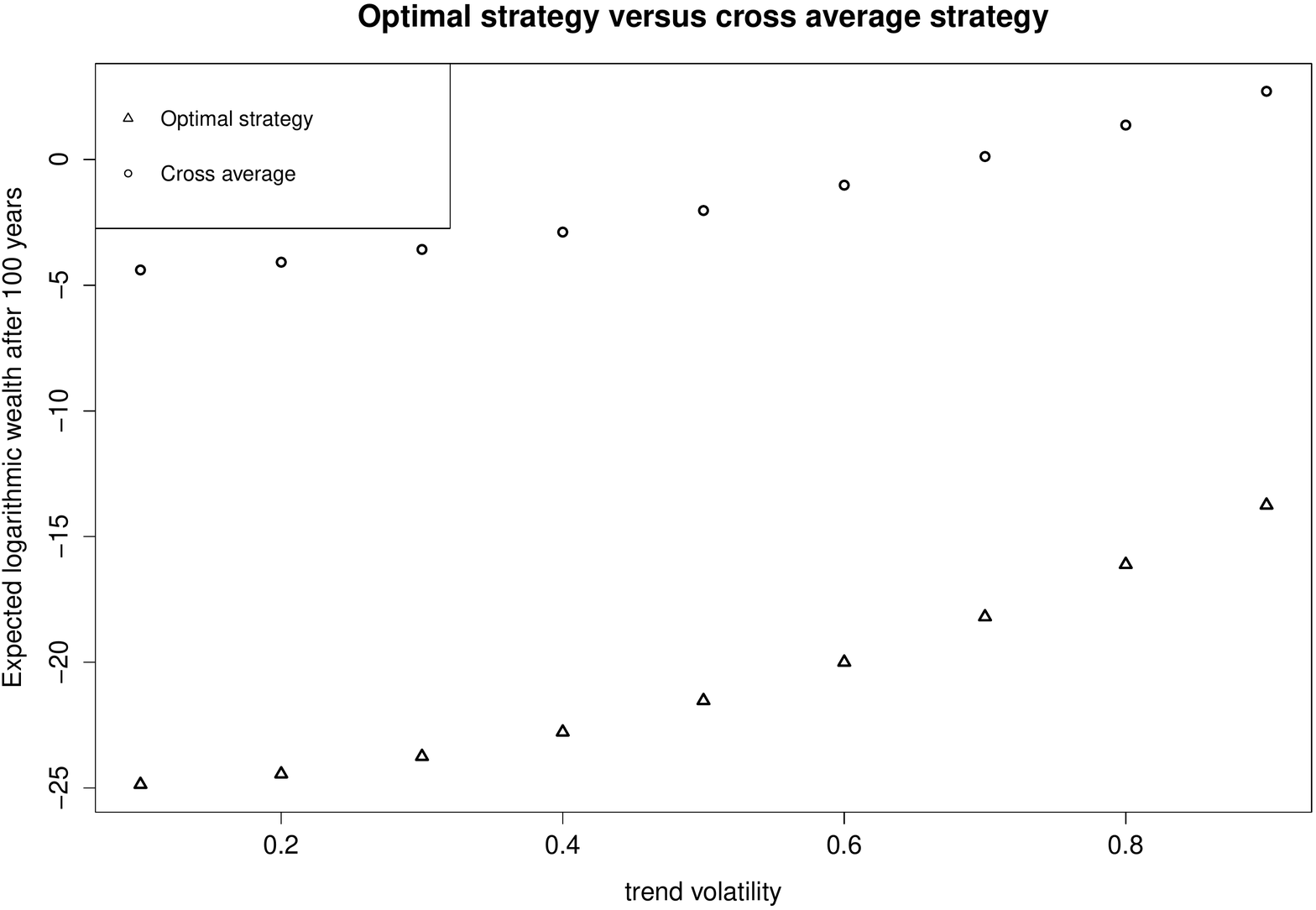}   
  \caption{The expected logarithmic returns of the optimal strategy (with $\tau=252$ days) and of the cross average strategy ($L_1=5$ days and $L_1=252$ days) as functions of $\sigma_\mu$ with $\lambda=4$, $\sigma_S=30\%$ and $T=100$ years}
\end{center}\label{Figu38}
\end{figure}
\subsubsection{Stability of the performances over several theoretical regimes under Heston's stochastic volatility model}
\paragraph{Model and optimal strategy}
The aim of this subsection is to check if the cross average strategy is more robust than the optimal trading strategy under Heston's stochastic volatility model (see \cite{Heston1} or \cite{Heston2} for details). To this end, consider a financial market living on a stochastic basis $( \Omega , \mathcal{G} ,\mathbf{G}, \mathbb{P} )$, where  $\mathbf{G}=\left\lbrace  \mathcal{G}_{t}, t \geqslant 0 \right\rbrace$ is the natural filtration associated to a three-dimensional Wiener process $(W^{S},W^{\mu},W^{V})$, and $\mathbb{P}$ is the objective probability measure. The dynamics of the risky asset $S$ is given by
\begin{eqnarray}\label{Model3VolSto}
\frac{dS_{t}}{S_{t}}&=&\mu_{t}dt+\sqrt{V_{t}}dW_{t}^{S},\\
d\mu_{t}&=&-\lambda\mu_{t}dt+\sigma_{\mu}dW_{t}^{\mu}\label{Model23VolSto},\\
dV_{t}&=&\alpha\left(V_{\infty}-V_{t} \right) dt+ \epsilon \sqrt{V_{t}} dW_{t}^{V}
\end{eqnarray}
with $\mu_{0}=0$, $V_{0}>0$, $d\left\langle W^{S},W^{\mu} \right\rangle_t=0$, and $d\left\langle W^{S},W^{V} \right\rangle_t=\rho dt$. We also assume that $\left( \lambda,\sigma_{\mu}\right) \in  \mathbb{R}_+^{*}\times\mathbb{R}_+^{*}$ and that $2 k V_{\infty}>\epsilon$ (in this case, the variance $V$  cannot reach zero and is always positive, see \cite{coxHeston} for details). Denote by $\mathbf{G}^{S}=\left\lbrace  \mathcal{G}_{t}^S \right\rbrace$ be the natural filtration associated to the price process $S$. In this case, the process $V$ is $G^S$-adapted. Now, assume that the agent aims to maximize his expected logarithmic wealth (on an admissible domain $\mathcal{A}$, which represents all the $\mathbf{G}^{S}$-progressive and measurable processes). In this case, his optimal portfolio is given by (see \cite{OptimalStrat}):
\begin{eqnarray*}
\frac{dP_{t}}{P_{t}}&=& \frac{E\left[  \mu_t | \mathcal{G}^S_t\right]}{V_{t}} \frac{dS_{t}}{S_{t}},\\
P_0&=& x.
\end{eqnarray*}
Let $\delta$ be a discrete time step, and denote by the subscript $k$ the value of a process at time $t_k = k \delta$. Using the scheme that produces the smallest discretization bias for the variance process (see \cite{Scheme} for details), the discrete time model is:
\begin{eqnarray}\label{Equation KalmanVolSto}
y_{k+1}=\frac{S_{k+1}-S_{k}}{\delta S_{k}}&=& \mu_{k+1} + u_{k+1},\\
\mu_{k+1}&=&e^{-\lambda \delta} \mu_{k}+v_{k},\label{Equation Kalman2VolSto}\\
V_{k+1}&=& V_{k}+\alpha\left(V_{\infty}-V_{k}^{+} \right) \delta +\epsilon \sqrt{V_{k}^{+}} z_{k} \label{Equation Kalman3VolSto}
\end{eqnarray} 
where $x^{+}=\max\left( 0,x\right), $ $u_{k+1}\sim \mathcal{N}\left( 0,\frac{V_{k}}{\delta}\right)$, $v_{k}\sim \mathcal{N}\left( 0,\frac{\sigma_{\mu}^{2}}{2\lambda}\left( 1- e^{-2 \lambda \delta} \right) \right)$ and $z_{k}\sim \mathcal{N}\left( 0,\delta\right)$. 
\paragraph{Monte Carlo simulations}
In this section, Monte Carlo simulations are used to check if the cross average strategy is more robust than the optimal trading strategy under Heston's stochastic volatility model. To this end, we consider the discrete model (\ref{Equation KalmanVolSto})-(\ref{Equation Kalman2VolSto})-(\ref{Equation Kalman3VolSto}) and we assume that $\alpha=4$ (quarterly mean-reversion of the variance process), that $\epsilon=5\%$, that $V_{\infty}=V_{0}=0.3^2$ (which means an initial and a long horizon spot volatility equal to $30\%$) and that $\rho=-60\%$ (when the spot decreases, the volatility increases). Moreover, we consider an investment horizon equal to 50 years and $\delta=1/252$ (which means that that a year contains 252 days and that each allocation is made daily). With this set-up, we consider several trend regimes, we simulate $M$ paths of the risky asset over 50 years and we implement two strategies:
\begin{enumerate}
\item The discrete time version of the optimal strategy presented above. Since the process $V$ is $G^S$-adapted, $V_k$ is observable at time $t_k$ and the conditional expectation of the trend is tractable with the non stationary discrete time Kalman filter (see \cite{KalmanBucy}). We assume that the agent thinks that the parameters are equal to $\lambda^a=1$ and $\sigma_{\mu}^a=90\%$ when he uses the Kalman filter.
\item The cross moving average strategy (introduced in section \ref{CrossAvSection}) with $\left(L_1,L_2 \right)=\left(5 \text{ days},252 \text{ days}\right) $  and the following allocation:
\begin{eqnarray*}
\theta_k =-1+2\indicator{G^d\left(k,L_1\right)>G^d\left(k,L_2\right)},
\end{eqnarray*}
where $G^d\left(k,L\right)$ is the discrete geometric moving average computed on the last $L$ values of $S$.
\end{enumerate}
The figures \hyperref[Figu39]{9} and \hyperref[Figu40]{10} represent the estimated performances of these strategies after 50 years as a function of the trend volatility $\sigma_\mu$ with $M=10000$ and respectively with $\lambda=1$ and 2. These results confirm that the performance of the  cross average strategy is less sensitive to a trend regime variation than the performance optimal trading strategy with parameters mis-specification.
Moreover, The figures \hyperref[Figu41]{11}, \hyperref[Figu42]{12}, \hyperref[Figu43]{13} and \hyperref[Figu44]{14} represent the empirical distribution of the logarithmic return of these strategies after 50 years over $M=10000$ paths for different configurations. These figures show that, even with a good calibration, the logarithmic return of the cross average strategy is less dispersed than the logarithmic return of the optimal strategy. Then the cross average strategy is more robust than the optimal strategy.
\begin{figure}[H]
\begin{center}
   \includegraphics[totalheight=7cm]{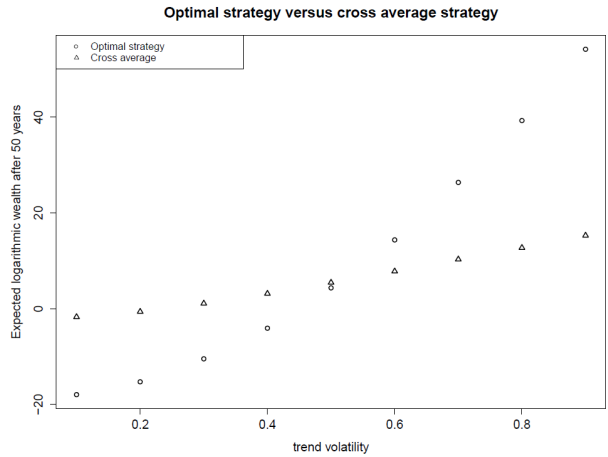}   
  \caption{The expected logarithmic returns of the optimal strategy (with $\lambda^a=1$ and $\sigma_{\mu}^a=90\%$) and of the cross average strategy ($L_1=5$ days and $L_1=252$ days) as functions of $\sigma_\mu$ with $M=10000$, $\lambda=1$, $\alpha=4$, $\epsilon=5\%$, $V_{\infty}=V_{0}=0.3^2$, $\rho=-60\%$  and $T=50$ years}
\end{center}\label{Figu39}
\end{figure}
\begin{figure}[H]
\begin{center}
   \includegraphics[totalheight=7cm]{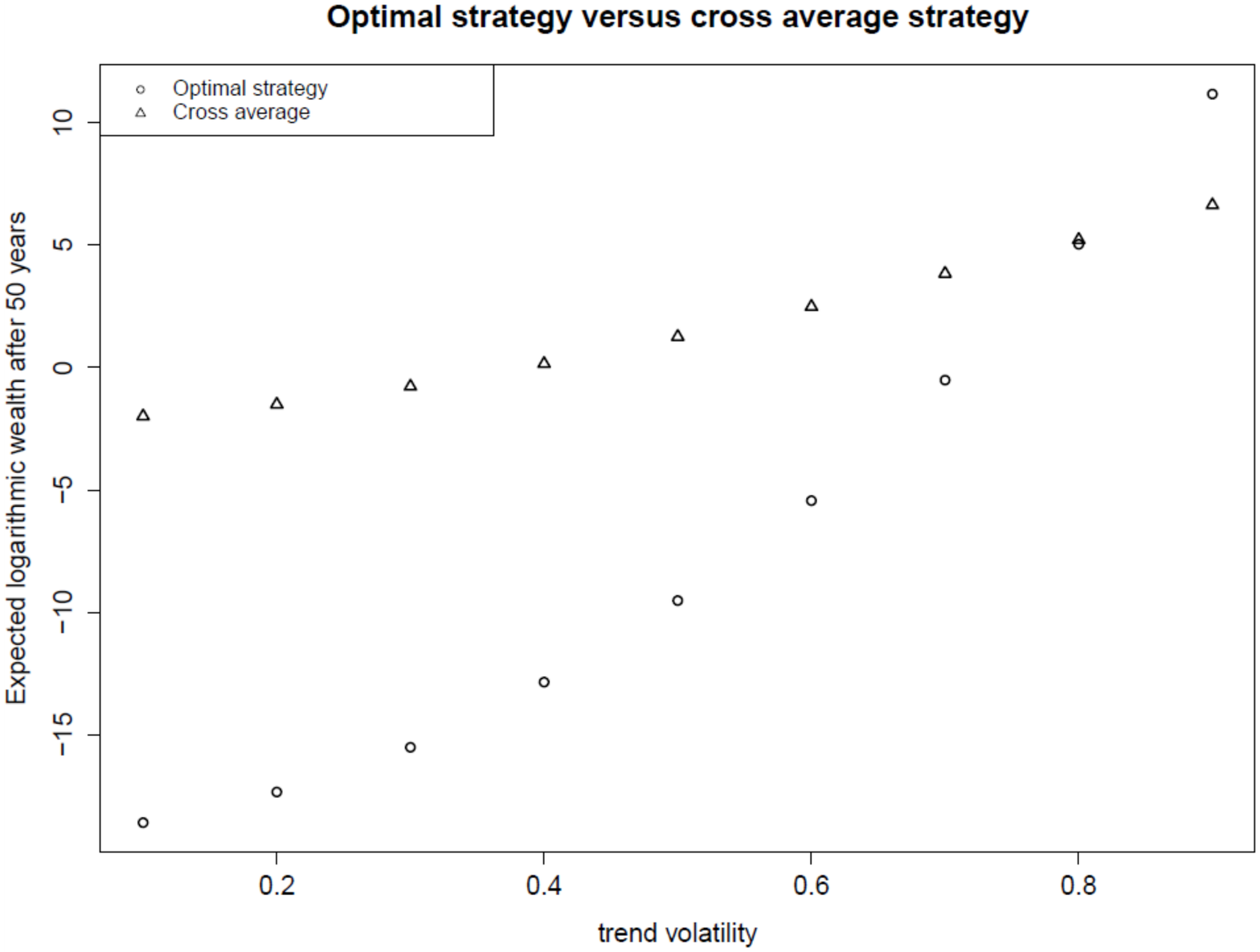}   
  \caption{The expected logarithmic returns of the optimal strategy (with $\lambda^a=1$ and $\sigma_{\mu}^a=90\%$) and of the cross average strategy ($L_1=5$ days and $L_1=252$ days)as functions of $\sigma_\mu$ with $M=10000$, $\lambda=2$, $\alpha=4$, $\epsilon=5\%$, $V_{\infty}=V_{0}=0.3^2$, $\rho=-60\%$  and $T=50$ years}
\end{center}\label{Figu40}
\end{figure}

\begin{figure}[H]
\begin{center}
   \includegraphics[totalheight=7.cm]{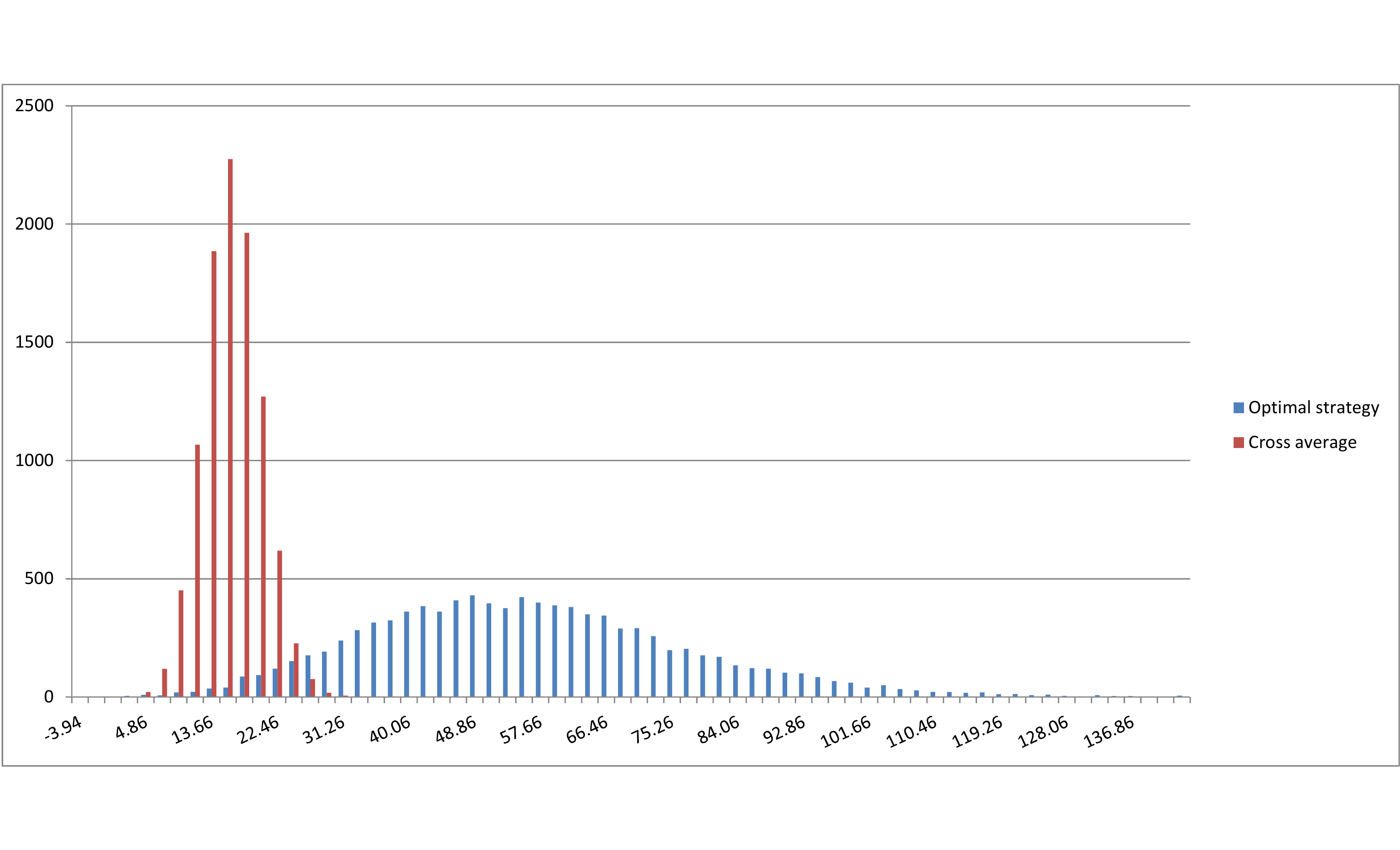}   
  \caption{Empirical distribution of the logarithmic return of the optimal strategy (with $\lambda^a=1$ and $\sigma_{\mu}^a=90\%$) and of the cross average strategy ($L_1=5$ days and $L_1=252$ days) with  $M=10000$, $\sigma_\mu=90\%$ , $\lambda=1$, $\alpha=4$, $\epsilon=5\%$, $V_{\infty}=V_{0}=0.3^2$, $\rho=-60\%$  and $T=50$ years}
\end{center}\label{Figu41}
\end{figure}
\begin{figure}[H]
\begin{center}
   \includegraphics[totalheight=7.cm]{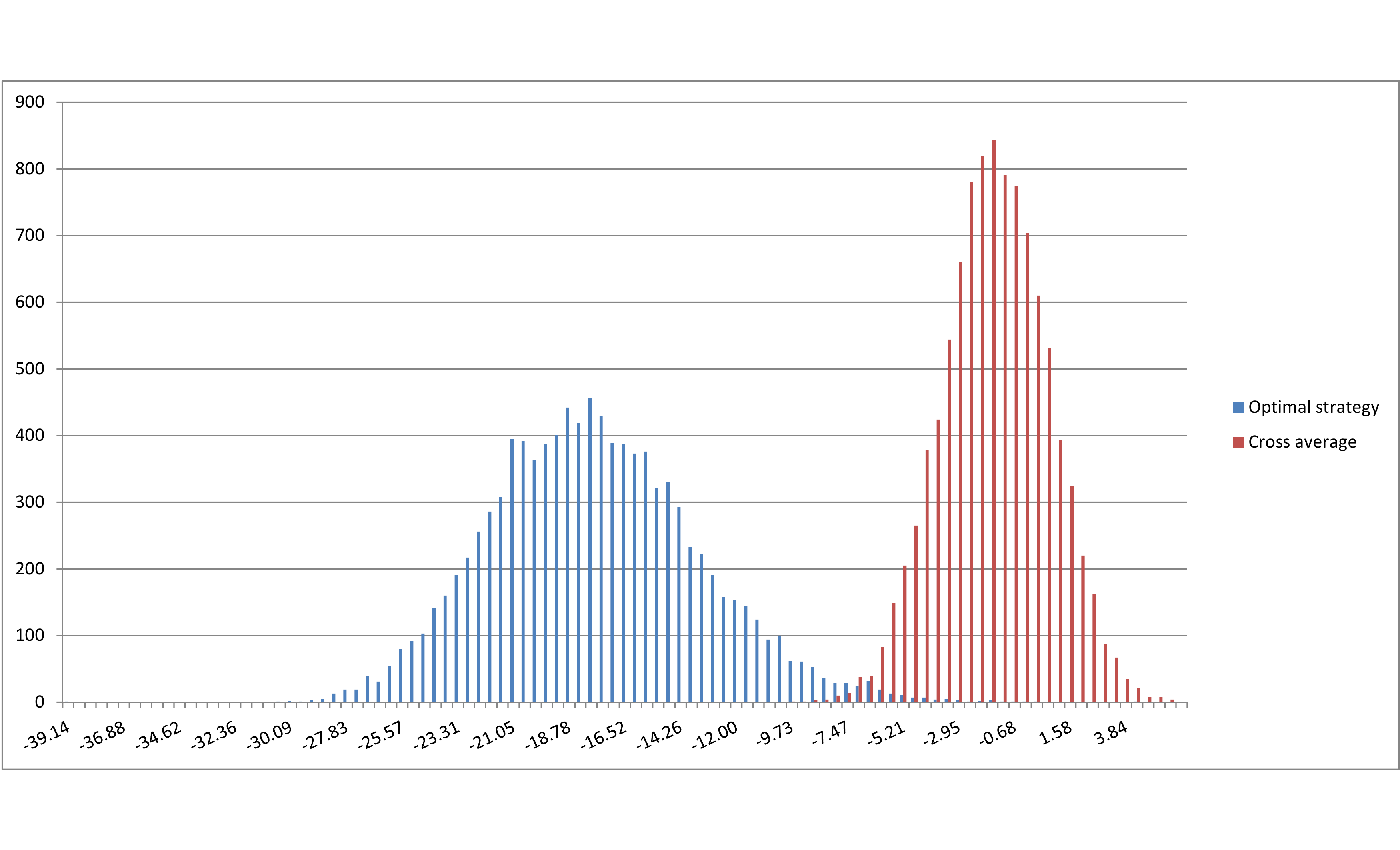}   
  \caption{Empirical distribution of the expected logarithmic return of the optimal strategy (with $\lambda^a=1$ and $\sigma_{\mu}^a=90\%$) and of the cross average strategy ($L_1=5$ days and $L_1=252$ days) with  $M=10000$, $\sigma_\mu=10\%$ , $\lambda=1$, $\alpha=4$, $\epsilon=5\%$, $V_{\infty}=V_{0}=0.3^2$, $\rho=-60\%$  and $T=50$ years}
\end{center}\label{Figu42}
\end{figure}
\begin{figure}[H]
\begin{center}
   \includegraphics[totalheight=7.cm]{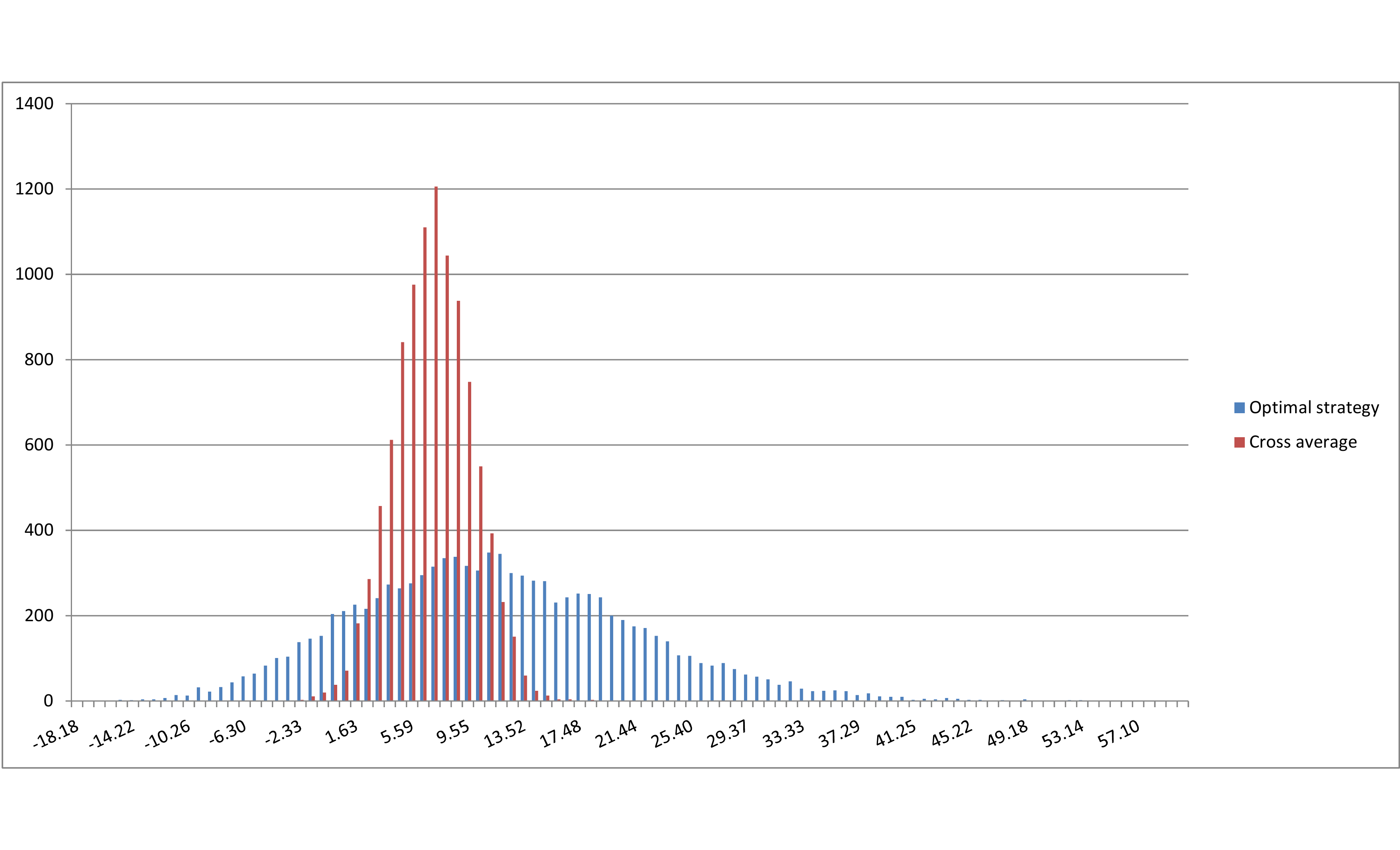}   
  \caption{Empirical distribution of the expected logarithmic return of the optimal strategy (with $\lambda^a=1$ and $\sigma_{\mu}^a=90\%$) and of the cross average strategy ($L_1=5$ days and $L_1=252$ days) with  $M=10000$, $\sigma_\mu=90\%$ , $\lambda=2$,  $\alpha=4$, $\epsilon=5\%$, $V_{\infty}=V_{0}=0.3^2$, $\rho=-60\%$  and $T=50$ years}
\end{center}\label{Figu43}
\end{figure}
\begin{figure}[H]
\begin{center}
   \includegraphics[totalheight=7.cm]{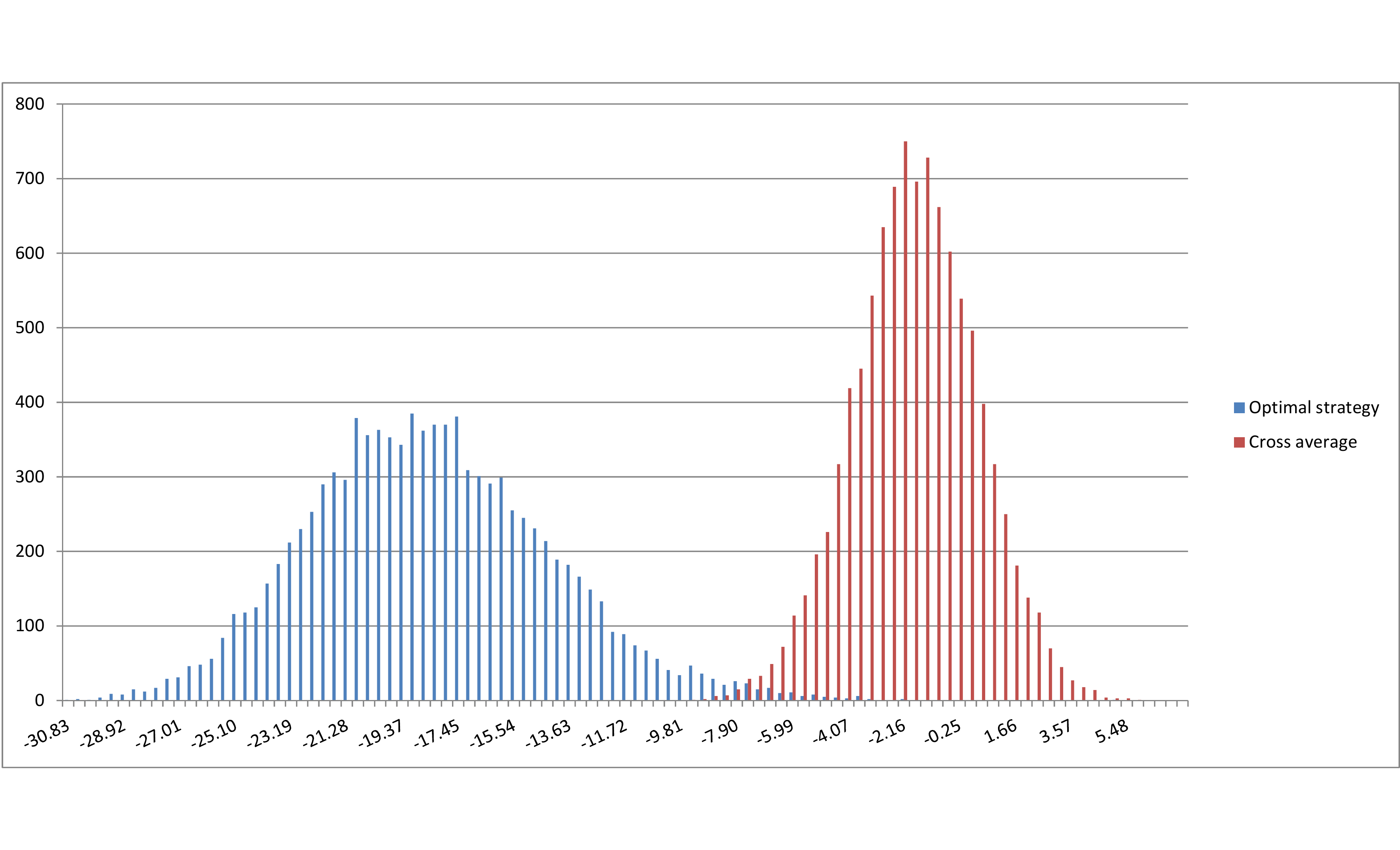}   
  \caption{Empirical distribution of the expected logarithmic return of the optimal strategy (with $\lambda^a=1$ and $\sigma_{\mu}^a=90\%$) and of the cross average strategy ($L_1=5$ days and $L_1=252$ days) with  $M=10000$, $\sigma_\mu=10\%$ , $\lambda=2$, $\alpha=4$, $\epsilon=5\%$, $V_{\infty}=V_{0}=0.3^2$, $\rho=-60\%$  and $T=50$ years}
\end{center}\label{Figu44}
\end{figure}
\subsubsection{Tests on real data}
Here we test the performances of the two previous strategies on real data. 
The performance of a strategy is evaluated with the annualised Sharpe ratio indicator (see \cite{Sharpe}) on relative daily returns. 
For the optimal strategy, we assume that $\tau=252$ business days, that $m=0.1$ (it has no impact on the Sharpe ratio indicator), and that the volatility $\sigma_S$ is computed over all the data and used since the beginning of the backtest. For the cross moving average strategy, we keep the same assumptions than the previous section (a  window of x days is replaced by a window of x business days).  
The universe of underlyings are nine stock indexes (the SP 500 Index, the Dow Jones Industrial average Index, the Nasdaq Index, the Euro Stoxx 50 Index, the Cac 40 Index, the Dax Index, the Nikkei 225 Index, the Ftse 100 Index and the Asx 200 Index) and nine forex exchange rates (EUR/CNY, EUR/USD, EUR/JPY, EUR/GBP, EUR/CHF, EUR/MYR, EUR/BRL, EUR/AUD and EUR/ZAR). The  period considered is from 12/22/1999 to 2/1/2015. In this test, we assume that these indexes are tradable and that the traded price is given by the closing price of the underlying. The backtest is done without transaction costs. For each strategy, the reallocation is made on a daily frequency. The figure \hyperref[Figu45]{15} gives the measured annualised Sharpe ratio of the 18 underlyings for each strategy. We observe that, even with an over-fitted volatility for the optimal strategy, the cross moving average strategy outperforms the optimal strategy except for the EUR/BRL. 
\begin{figure}[H]
\begin{center}
   \includegraphics[totalheight=10cm]{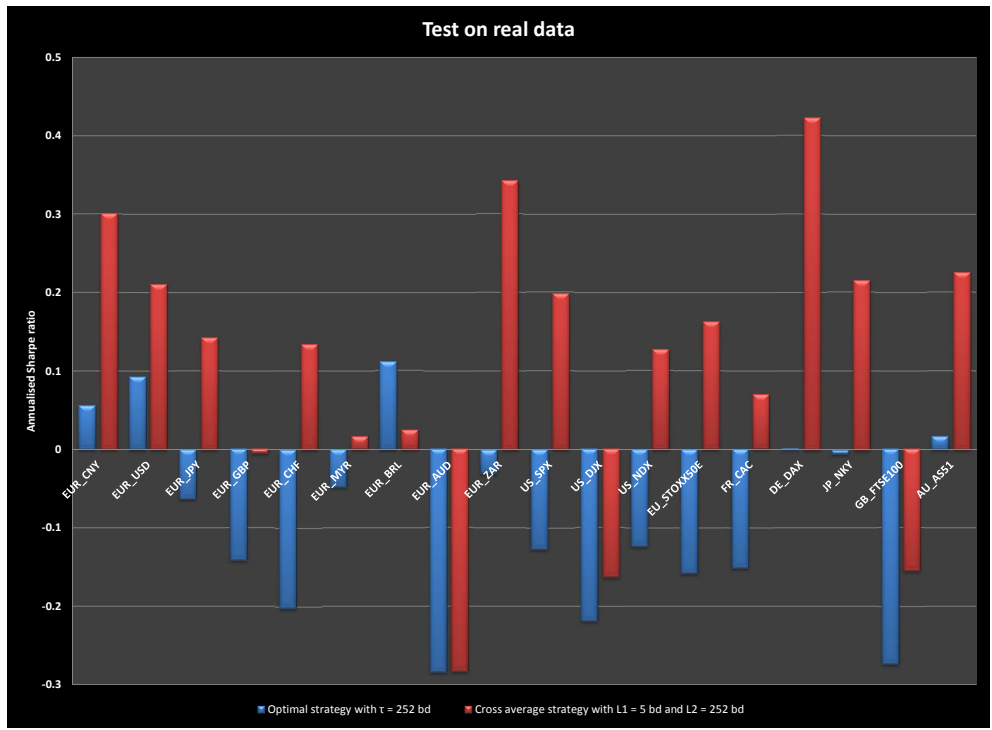}   
  \caption{Sharpe ratio of the optimal strategy (with $\tau=252$ bd) and of the cross average strategy ($L_1=5$ bd and $L_1=252$ bd) on real data from 12/22/1999 to 2/1/2015}
\end{center}\label{Figu45}
\end{figure}
\newpage
\section{Conclusion}
The present work quantifies the performances of the optimal strategy under parameters mis-specification and of a cross moving average strategy using geometric moving averages with a model based on an unobserved mean-reverting diffusion. 

For the optimal strategy, we show that the asymptotic expectation of the logarithmic returns is a an increasing function of the signal-to-noise ratio and a decreasing function of the trend mean reversion speed.

We find that, under parameters mis-specification, the performance can be positive under some conditions on the model and strategy parameters. Under the same assumptions, we show the existence of an optimal duration which is equal to the Kalman filter duration if the parameters are well-specified. 

For the cross moving average strategy, we also provide the asymptotic logarithmic return of this strategy as a function of the model parameters.

Moreover, the simulations show that, with a model based on an unobserved mean-reverting diffusion, and even with a stochastic volatility, technical analysis investment is more robust than the optimal trading strategy. The empirical tests on real data confirm this conclusion. 

\newpage
\bibliographystyle{authordate1}
\bibliographystyle{apalike}
\bibliography{BibliFinalPhd}{}

\end{document}